\documentclass[lettersize,journal]{IEEEtran}
\IEEEoverridecommandlockouts
\usepackage[english]{babel}
\usepackage{cite}
\usepackage{amsthm,amsmath,amssymb,amsfonts}

\usepackage{graphicx}
\usepackage{textcomp}
\usepackage{xcolor}
\usepackage{multirow}
\usepackage{graphicx}
\usepackage{subcaption}
\usepackage{pgf-umlsd}
\usepackage{comment} 
\usepackage{url}
\usepackage{breqn}
\newcommand{\Lim}[1]{\raisebox{0.5ex}{\scalebox{0.8}{$\displaystyle \lim_{#1}\;$}}}
\usepackage[linesnumbered,ruled,vlined]{algorithm2e}
\SetKwInput{KwInput}{Input}
\SetKwInput{KwOutput}{Output}

\renewcommand*{\algorithmcfname}{Protocol}
\newtheorem{theorem}{Theorem}
\newtheorem{proposition}[theorem]{Proposition}
\begin{document}

\title{Enhanced Security and Privacy via Fragmented Federated Learning}

\author{Najeeb~Moharram~Jebreel,           Josep~Domingo-Ferrer,~\IEEEmembership{Fellow,~IEEE,}
Alberto~Blanco-Justicia and
David~S\'anchez
\thanks{The authors are with the UNESCO Chair in Data Privacy, CYBERCAT Center
for Cybersecurity Research of Catalonia, Department of
Computer Engineering and Mathematics, Universitat Rovira i
Virgili, Av. Pa\"{\i}sos Catalans 26, E-43007 Tarragona, Catalonia, e-mail
\{najeeb.jebreel, josep.domingo, alberto.blanco, david.sanchez\}@urv.cat.}}

\markboth{Journal of \LaTeX\ Class Files,~Vol.~14, No.~8, August~2015}%
{Shell \MakeLowercase{\textit{et al.}}: Bare Demo of IEEEtran.cls for Computer Society Journals}


\maketitle

\begin{abstract}
In federated learning (FL), a set of participants share updates computed on their local data with an aggregator server that combines updates into a global model. However, reconciling accuracy with privacy and security is a challenge to FL. On the one hand, good updates sent by honest participants may reveal their private local information, whereas poisoned updates sent by malicious participants may compromise the model's availability and/or integrity. On the other hand, enhancing privacy via update distortion damages accuracy, whereas doing so via update aggregation damages security because it does not allow the server to filter out individual poisoned updates. To tackle the accuracy-privacy-security conflict, we propose {\em fragmented federated learning} (FFL), in which participants randomly exchange and mix fragments of their updates before sending them to the server. 
To achieve privacy, we design a lightweight protocol that allows participants to privately exchange and mix encrypted fragments of their updates so that the server can neither obtain individual updates nor link them to their originators. 
To achieve security, we design a reputation-based defense tailored for FFL that builds trust in participants and their mixed updates based on the quality of the fragments they exchange and the mixed updates they send. Since the exchanged fragments' parameters keep their original coordinates and attackers can be neutralized, the server can correctly reconstruct a global model from the received mixed updates without accuracy loss. Experiments on four real data sets show that FFL can prevent semi-honest servers from mounting privacy attacks, 
can effectively counter poisoning attacks and can  keep the accuracy of the global model.
\end{abstract}

\begin{IEEEkeywords}
Federated learning, Privacy, Security, Accuracy, Fragmentation.
\end{IEEEkeywords}



\section{Introduction}\label{sec:introduction}

\IEEEPARstart{F}{ederated learning} (FL, \cite{mcmahan2017communication}) enables multiple participants to jointly train a machine learning (ML) model without sending their data to a central server. In FL, participants compute model updates based on their local data and a global model received from a coordination server, to whom they send the resulting updates. The server aggregates the received updates to obtain a new global model, which is redistributed among the participants. This collaboration entails mutual benefits for all parties: i) the participants and the server obtain more accurate models due to learning from joint training data, ii) participants keep their private local data in their devices, and iii) the server distributes the computational load of training across the participants' devices ({\em e.g.}, smartphones) \cite{bonawitz2019towards}. Since training data do not leave the participants' devices, FL is a suitable option for 
scenarios dealing with personal data, such as facial recognition~\cite{xu2017end}, voice assistants~\cite{bhowmick2018protection}, healthcare~\cite{brisimi2018federated}, next-word prediction~\cite{hard2018federated} and intrusion detection in IoT
networks~\cite{mothukuri2021federated}, or in case data collection and processing are restricted due to privacy protection laws such as the GDPR~\cite{voigt2017eu}.
Despite these advantages, FL is vulnerable to privacy and security attacks~\cite{kairouz2019advances, mothukuri2021survey}.

Regarding privacy, several works \cite{mcmahan2017learning, nasr2019comprehensive, melis2019exploiting} have demonstrated that a \emph{semi-honest} server can analyze individual updates to infer sensitive information on a participant's local training data. Recent powerful privacy attacks show that
it is possible to reconstruct the original training data by inverting the gradients of updates \cite{zhu2020deep, zhao2020idlg, geiping2020inverting}.

Regarding security, FL is vulnerable to poisoning attacks~\cite{blanco2021achieving}.
Since the server has no control over the behavior of participants, any of them may deviate from the prescribed training protocol to attack the model by conducting either untargeted poisoning ({\itshape i.e.} Byzantine) attacks~\cite{blanchard2017machine, wu2020federated} or targeted poisoning attacks~\cite{biggio2012poisoning, fung2020limitations, bagdasaryan2020backdoor}. 
In the former type of attacks, the attacker aims to degrade the model's overall performance, whereas in the latter, he aims to cause the global model to misclassify some attacker-chosen inputs.

Some solutions have been proposed to prevent the server from analyzing individual updates or linking them to their originators. These involve well-known privacy-enabling methods: differential privacy (DP) \cite{dwork2014algorithmic}, homomorphic encryption (HE) \cite{gentry2009fully} and secure multiparty computation (SMC) \cite{yao1982protocols}. 
DP-based methods protect the participants' data by injecting random noise into the parameters of updates at the cost of sa\-cri\-ficing the accuracy of the global model \cite{domingo2020limits,acmdp}. 
On the other hand, HE- and SMC-based methods securely aggregate the updates of participants before sending them to the server. 
Yet, HE and SMC have a high computational cost 
and they prevent the server from inspecting individual updates, which makes approaches based on these techniques 
vulnerable to security attacks. 
Note that countermeasures against security attacks require direct access by the server 
to the individual updates in order to detect and/or filter out those that are poisoned 
\cite{yin2018byzantine, blanchard2017machine, fung2018mitigating}.

Simultaneously achieving privacy, security and accuracy is a tough challenge for FL \cite{li2020federated, kairouz2019advances, blanco2021achieving}.
In fact, this is one of several challenges in ML due to contradicting requirements; see~\cite{liu2020machine} about other conflicts.

Our goal is to address the following puzzle: \emph{``Can we prevent a semi-honest server from performing privacy attacks on individual updates while learning an accurate global model and ensuring protection against security attacks?''}. 

To this end, we propose \emph{fragmented federated learning} (FFL), a framework in which participants randomly exchange fragments of updates among them before sending them to the server. Our work brings the following contributions:
\begin{itemize}
    \item We propose a novel lightweight protocol that i) allows participants to privately exchange and mix random fragments of their updates, ii) enables the server to correctly aggregate the global model from the mixed updates, and iii) prevents the server from recovering the complete original updates or linking them to their originators.
    \item We propose a new reputation-based defense tailored to FFL against security attacks. 
    Specifically, the server selects participants for training and adaptively aggregates their mixed updates according to their global reputations. 
    Also, honest participants do not exchange fragments with participants with low local reputations.
    Reputations are computed based on the quality of the updates
    the participants send and the fragments they exchange. 
    \item We provide extensive theoretical and empirical analyses to assess the accuracy, privacy and security offered by FFL, and we quantify the computation overhead and communication cost it incurs.
\end{itemize}
The rest of this paper is organized as follows. 
Section~\ref{sec:related} discusses related works.
Section~\ref{sec:prelim} introduces preliminary notions. Section~\ref{sec:attack_model} describes the attacks being considered.
Section~\ref{sec:meth} presents the fragmented federated learning framework. 
Section~\ref{sec:privacy} and Section~\ref{sec:security} provide privacy and security analyses of FFL.
Section~\ref{sec:setup} details the experimental setup and evaluates our approach w.r.t. accuracy, robustness against attacks and runtime.
Convergence and complexity are analyzed in Section~\ref{sec:convergence} and Section~\ref{sec:complexity}.
Section~\ref{sec:conclusion} gathers conclusions and proposes several lines of future research.

\section{Related Work}
\label{sec:related}

Most works in the FL literature tackle either robustness versus attacks against privacy or robustness versus attacks against security (poisoning attacks), but they do not consider both types of robustness simultaneously. 

\textbf{Private FL.}
On the privacy side, several works have been proposed to prevent the server from seeing individual updates. \cite{bonawitz2017practical, so2021turbo} use secret sharing to hide individual updates. \cite{aono2017privacy, hardy2017private, zhang2020batchcrypt} encrypt the local updates and use homomorphic encryption to compute the global model from the encrypted updates. However, both approaches are vulnerable to poisoning attacks because they hinder the analysis of individual updates.
Other works adopt differential privacy (DP) \cite{mcmahan2017learning, bhowmick2018protection, geyer2017differentially}, which adds noise to local updates before sending them to the server. DP is practical but it only offers strong privacy guarantees for small values of $\epsilon$ that, due to the noise they add, significantly hamper the accuracy of the global model~\cite{domingo2020limits, acmdp}. 

\textbf{Secure FL.}
On the security side, several attack-resistant aggregation defenses have been proposed such as Multi-Krum~\cite{blanchard2017machine}, Bulyan~\cite{mhamdi2018hidden}, FoolsGold~\cite{fung2020limitations} and DRACO~\cite{chen2018draco}.
However, they require direct access to the original individual updates, which interferes with protecting 
the privacy of participants.
Although the median and the trimmed mean~\cite{yin2018byzantine} can be applied on mixed updates, they are not suitable for high-dimensional models because their estimation error scales up with the size of the model in a square-root manner~\cite{chang2019cronus}.

\textbf{Private and secure FL.}
Recently, the literature has witnessed a growing interest in achieving FL that is both privacy-preserving and secure.
\cite{naseri2020toward} uses DP to address both privacy and robustness against backdoor attacks. However, it does not deal with the trade-off between privacy and accuracy, and it does not consider poisoning attacks different from backdoor attacks. 
PEFL~\cite{privacy_enhanced9524709} tries to address both privacy and security, but assumes there are two \emph{non-colluding} servers that collaborate to filter out malicious updates while preventing each other from seeing individual updates. 
Moreover, PEFL builds on linear homomorphic encryption and a packing technique, and it involves exchanging the encrypted updates in four interacting protocols between the two servers for filtering and aggregation, which causes high communication and computation overheads. 
ShieldFL~\cite{ma2022shieldfl} also assumes two \emph{non-colluding} servers, and uses a two-trapdoor HE scheme based on the Paillier cryptosystem to achieve both secure and privacy-preserving FL.
In ShieldFL, the two servers execute three interactive protocols to compute the cosine similarity of the local updates in ciphertext. They then use the computed cosine similarities to filter out potential poisoned updates. Unfortunately, this also imposes significant computation and communication costs on the participating entities.
BREA~\cite{so2020byzantine} proposes a single-server framework where each participant secret-shares her local update with all the other participants in the system. 
Also, each participant locally computes the participant-wise Euclidean distances to the shares of all participants, and then sends the computed distances to the server, which uses them to filter out potential poisoned updates. 
The server then selects the potential good updates and asks the participants to locally aggregate the selected good shares and upload the aggregates to the server. Finally, the server reconstructs the global model from the received aggregates and sends it to the participants in the next training round. 
Although this work has the advantage of being single-server, it imposes high computation and communication overheads on the participating parties. 
\cite{ma2022differentially} integrate local DP on the participant's side with an intermediate shuffler between the participants and the aggregator server to achieve privacy.  
Besides, they use a Byzantine-robust stochastic aggregation algorithm at the server side to achieve security.
However, this work is subject to the inevitable trade-off between privacy and accuracy due to the use of DP. 
In addition, the shuffler is a third party, and hence its honesty is not guaranteed.
SAFELearning \cite{zhang2021safelearning} is based on the work of \cite{bonawitz2017practical} to support backdoor detection and privacy-preserving aggregation simultaneously. To this end, it randomly divides participants into subgroups, securely aggregates a sub-model for each subgroup and filters out malicious sub-models instead of individual models. However, this approach faces a trade-off of another kind: the smaller the number of participants in a subgroup, the less privacy for these participants; conversely, the larger the number of participants, the easier it is for the attacker to hide her malicious model amid honest ones. The paper proposes to disclose part of the parameters of the aggregated sub-models, leading to another privacy/security trade-off. The authors of~\cite{domingo2021secure} propose a co-utile FL to solve the accuracy-privacy-security conflict. The proposed solution preserves the global model accuracy and allows defending against security attacks, but its privacy is limited to only breaking the link between local updates and their originators. Although this provides a level of privacy protection, a semi-honest server still has direct access to original unlinked updates. Hence, it can use them to perform several privacy attacks, such as membership inference attacks and reconstruction attacks. ~\cite{chen2020training, zhang2021shufflefl} employ Trusted Execution Environments (TEEs) on the participants (for local training) and on the servers (for secure aggregation) in order to achieve accurate and privacy-preserving FL. However, the limited memory size of current TEEs makes them impractical for large DL models or large-scale FL systems.
Besides, the authors of those papers do not consider data poisoning attacks such as label-flipping attacks, and assume trust in the manufacturers of the TEEs, which seems too strong an assumption.

\section{Preliminaries}
\label{sec:prelim}

\subsection{Deep neural networks}
A deep neural network (DNN) is a function $F(x)$, obtained by composing $L$ functions $f_l, l\in [1, L]$, that maps an input $x\in \mathbb{R}^n$ to an output $y\in \mathbb{R}^z$. Each $f_l$ is a layer that is parameterized by a weight matrix $w_l$, a bias vector $b_l$ and an activation function $\sigma_l$. 
In this paper we use predictive DNNs as $z$-class classifiers. 

\subsection{Federated learning}
In federated learning (FL), $K$ participants and an aggregator server $A$ collaboratively build a global model $W$. In each training round $t\in[1, T]$,
the aggregator randomly selects a subset of participants $S$ of size $n = C \cdot K \ge 1$ where $K$ is the total number of participants in the system, and $C$ is the fraction of participants that are selected in $t$. After that, $A$ distributes the current global model $W^{t}$ to all participants in $S$. Besides $W^{t}$, $A$ sends a set of hyper-parameters to be used to train the local models, which includes the number of local epochs $E$, the local batch size $BS$, and the learning rate $\eta$. After receiving $W^t$, each participant $k$ divides her local data into batches of size $BS$ and performs $E$ local training epochs on her data to compute her local update $W^{t+1}_{k}$. Finally, participants upload their updates to $A$, which aggregates them into a new global model $W^{t+1}$.
The federated averaging algorithm (\textit{FedAvg})\cite{mcmahan2017communication} is usually employed to perform the aggregation, and it is defined as 
\begin{equation}
\label{eq:fedavg}
   W^{t+1} = \sum_{k=1}^{K} \frac{d_{(k)}}{d} W^{t+1}_{(k)},
\end{equation}
where $d_{(k)}$ is the number of data points locally held by worker $k$, and $d$ is the total number of data points locally held by the $K$ workers, that is, $d=\sum_{k=1}^{K} d_{(k)}$.
Note that FedAvg is the standard way to aggregate updates in FL and is not meant to counter security attacks.

\section{Attack Models}
\label{sec:attack_model}
\textbf{Privacy attack model.} We focus on a \emph{semi-honest} server $A^s$, who follows the protocol honestly, but tries to infer information about the private local data of by participants for training. Even though privacy attacks may also be orchestrated by participants based on the successive global models, the performance of such attacks is quite limited and degrades significantly as the number of participants increases \cite{melis2019exploiting}. Server-side attacks are much stronger, especially when the server sees local updates individually and can link them to their originators.  
In particular, membership inference attacks aim to determine if a specific example is part of the training data set \cite{melis2019exploiting, nasr2019comprehensive}. Property inference attacks try to infer specific properties about the training data \cite{ganju2018property, melis2019exploiting}. Distribution estimation attacks aim to obtain examples from the same distribution of the participants' training data \cite{hitaj2017deep}. 
Finally, reconstruction attacks are the most ambitious in that they attempt to extract the original training data from a participant's local update \cite{zhu2020deep, zhao2020idlg, geiping2020inverting}. To mount any of these attacks, the server needs to access individual local updates. Therefore, our goal is to disable privacy attacks by preventing the server from obtaining the participants' original updates.

\textbf{Security attack model.} Since the server has no control over the participants' behavior, a \emph{malicious} participant may deviate from the prescribed training protocol and attack the global model. Depending on the attacker's objective, security attacks can be divided into untargeted attacks \cite{fang2020local, biggio2012poisoning, jagielski2018manipulating}, and targeted attacks \cite{bagdasaryan2020backdoor, bhagoji2019analyzing, shafahi2018poison}. 
The former are attacks against model availability 
({\em e.g.} to prevent the model from converging), whereas the 
latter are against model integrity ({\em e.g.} to fool the global model into making incorrect predictions on some attacker-chosen inputs). Security attacks can be performed in two ways: model poisoning and data poisoning. In model poisoning \cite{blanchard2017machine, wu2020federated, bagdasaryan2020backdoor}, the attacker manipulates the model parameters before sending her update to the server. In data poisoning \cite{biggio2012poisoning, fung2018mitigating}, the attacker injects bad or biased data into her training data set before training her local model. 
In our work, we consider a number of attackers $K'\leq K/5$, that is, no more than $20\%$ of the $K$ participants in the system. 
Although some works in the literature assume larger percentages of attackers, finding more than $20\%$ of attackers in real-world FL scenarios is highly unlikely.
For example, with the millions of users~\cite{davenport_corbin} of Gboard~\cite{hard2018federated}, controlling even a small percentage of user devices requires the attacker(s) to compromise a large number of devices, which demands huge effort and resources and is therefore impractical. 
We assume the $K'$ attackers carry out two types of attacks: (1) untargeted attacks based on Gaussian noise \cite{li2019rsa, wu2020federated} and (2) targeted attacks based on label-flipping \cite{biggio2012poisoning,fung2018mitigating}. Furthermore, we assume that the attacker(s) have no control over the server or the honest participants.

\section{Fragmented Federated Learning}
\label{sec:meth}
In this section, we present the \emph{fragmented federated learning} (FFL) framework.
First, we give an overview of our framework and then present its design and protocols in detail. Table~\ref{tab:notation} summarizes the notation used in this paper.

\begin{table}[ht]
\centering
\scriptsize
\caption{Notation used in this paper}
\label{tab:notation}
\begin{tabular}{|l|l|}
\hline
Notation & Description \\ \hline
$W$                 &   Federated learning model or update \\
$W^t$ & Model or update at round $t$\\ 
$|W|$ & Number of model parameters\\
$\lambda$ & Bitlength of a parameter\\
$w_l$               &  Weight matrix of layer $l$ \\
$b_l$               &  Bias vector of layer $l$            \\
$\sigma_l$           &  Activation function of layer $l$ \\
$D$ & Dimensionality of the model\\
$L$ & Number of layers of the model\\
$D_L$ & Last-layer dimensionality\\
$K$                 &  Number of participants \\
$A$                 &  Aggregator server \\ 
$T$                 &   Number of training rounds \\
$C$                 &  Fraction of selected participants \\
$n$                 & Number of selected participants     \\
$S$                 &  Set of selected participants\\
$Sc$                & Set of candidate participants\\
$E$                 &  Number of local epochs \\
$BS$                & Size of local batch size \\ 
$\eta$              &   Local learning rate \\
$W_k$ & Participant $k$'s update\\
$W_{k,i}$           &  $i$-th parameter of $W_k$\\
$K'$                & Number of malicious participants     \\
$\gamma^t$            & Global reputation vector held by the server at time $t$\\
$\gamma_k^t$            & Global reputation given by the server to participant $k$ at time $t$ \\
$Q1_x$            & First quartile of the values of magnitude $x$ \\
$\zeta^t_k$         & Local reputation vector held by participant $k$ at time $t$\\
$\zeta_{k, j}^t$         & Local reputation given by participant $k$ to participant $j$ at time $t$\\
$p$ & Large prime (2048-bit long)\\
$\mathbb{G}_p$ & Multiplicative group of integers mod $p$\\
$g$ & Generator of $\mathbb{G}_p$\\
$PRNG(.)$           & Public pseudo-random number generator \\
$Enc_{{pk}_{A}}(\cdot)$ & Encryption under $A$'s public key \\ 
$Dec_{{sk}_{A}}(\cdot)$    & Decryption under $A$'s private key \\
$(W_k)_{mix}$ & Participant $k$'s mixed update\\
$(W_k)^{'}_{mix}$   &  Participant $k$'s encrypted mixed update\\
$s_{r_k}$           & Seed to generate $k$'s one-time pad (OTP) \\
$r_k$               & OTP used to encrypt $k$'s fragments \\
$m$                 & Binary mask: vector of 0's and 1's \\ 
$\lnot{m}$           & 1's complement of $m$ \\
\hline
\end{tabular}
\end{table}

\subsection{Overview}
\label{sec:overview}

 \begin{figure}[!ht]
    \centering
      \includegraphics[width=1\linewidth]{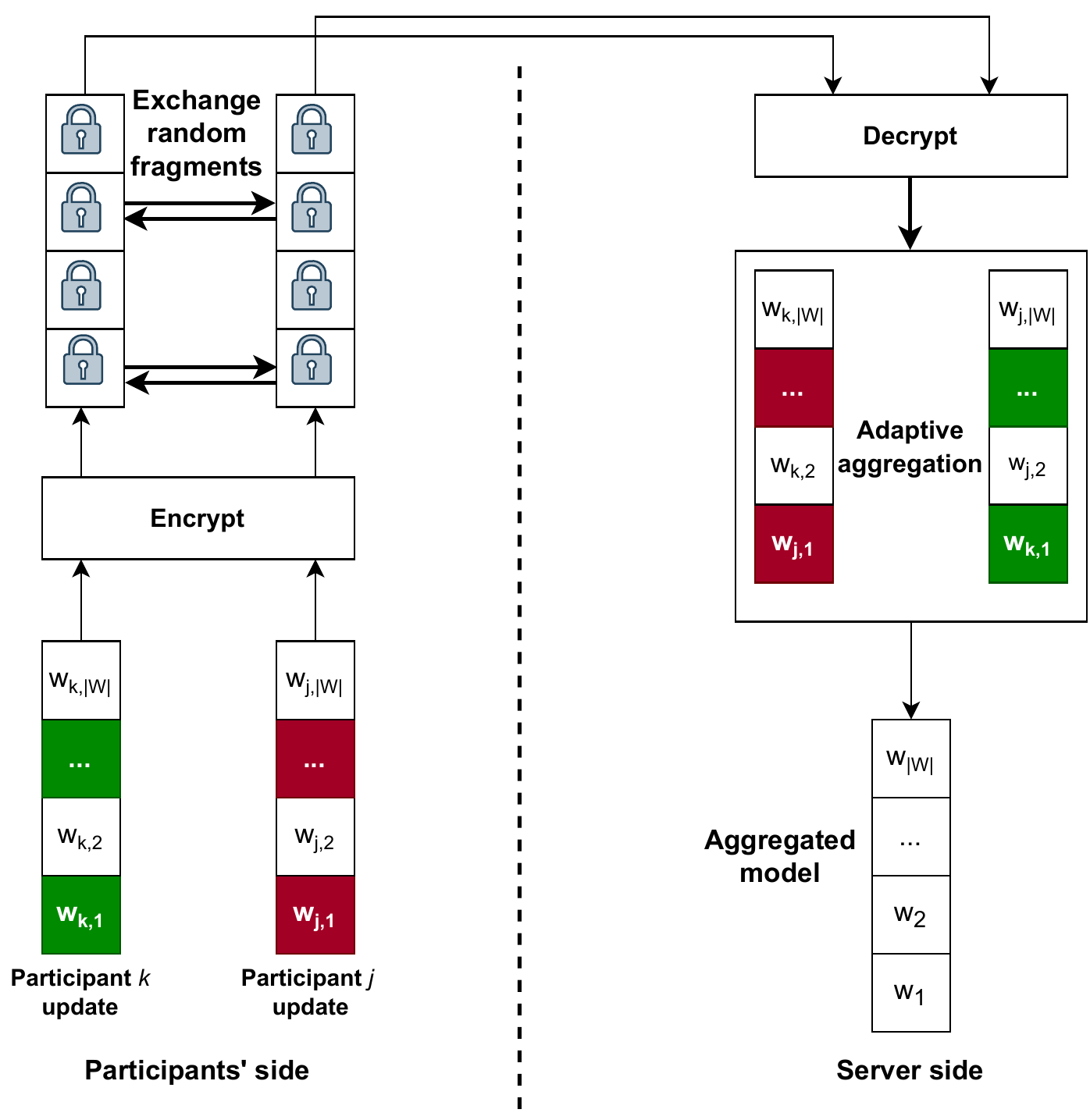}
      \caption{Overview of the FFL framework}
      \label{ffl_framework}
\end{figure}
Fig.\ref{ffl_framework} shows an overview of the FFL framework. The key idea is to have the participants randomly fragment and mix their updates before sending them to the server. Specifically, two participants agree on some symmetric random indices in their update vectors and exchange the parameters of those indices after they are encrypted. The participants then send the encrypted mixed updates to the server instead of their original updates. 
Since exchanging parameters is done without changing their original coordinate positions, the server can calculate the average of the mixed updates after decrypting the encrypted mixed updates and obtain the same updated global model that would result from averaging the original updates. 
Sending mixed updates instead of the original ones breaks the link between the updates and their originators, and does not give the server direct access to individual updates. Thus, FFL  prevents a semi-honest server from mounting powerful privacy attacks.
Also, the global model's accuracy is preserved because the parameters of the exchanged fragments are kept unaltered. 
Yet, exchanging fragments between participants should neither impose significant communication or computation overheads nor protect against server-side privacy attacks at the price of facilitating participant-side privacy attacks. Thus, we design a fragment exchanging protocol based on lightweight cryptographic tools to:
i) allow the participants to exchange and mix random fragments of their updates while incurring 
very minor communication and computation overheads,
ii) prevent participants from seeing each other's original fragments,
and iii) allow the server to correctly compute the updated global model from the mixed updates without being able to recover the individual original updates or link them to their originators. 
However, averaging updates to compute the global model goes against countering security attacks.
Also, the design of the exchange protocol gives $n'$ attackers the chance of poisoning $2n'$ coordinates.
We tackle both issues by designing a novel reputation-based defense tailored for FFL that builds trust in participants based on the quality of the mixed updates they send and the fragments they exchange.
Specifically, the server holds a global reputation vector and uses it to select participants for training and adaptively aggregate their mixed updates. A participant who repeatedly sends poisoned mixed updates will have a lower global reputation, and thus will not be selected in future training rounds.
Also, the mixed updates she sends will have little to no influence on the global model aggregation.
On the other hand, each participant holds a local reputation vector and uses it to decide whether to exchange fragments with the other participants.
The local reputation increases or decreases based on the quality of the fragments exchanged by participants.
An attacker who exchanges poisoned fragments with honest participants will eventually find no honest participant to exchange with.
The similarity between the mixed updates and their centroid is used to measure their quality.

\subsection{FFL design}
\label{sec:design}
There are $K + 1$ parties in the FFL framework: $K$ participants and an aggregator server $A$. 
We assume that $K'$ participants (with $K'\leq K/5$) may be malicious and attack the global model's availability or integrity by sending poisoned updates to $A$. 
Also, we assume that the aggregator server is semi-honest and may use the participants' updates to infer information about their local private data. 
Before starting any training task, $A$ generates a public/private key pair $(pk_A, sk_A)$ and broadcasts $pk_A$ to the $K$ participants.
Then, the server uses the metadata of the participants' devices (like IP or MAC addresses) to register them into the system. 
After that, $A$ assigns each participant a unique pseudonym. 
All parties are assumed to possess pairwise secure communication channels using communication protocols such as TLS or HTTPS and communicate with each other using pseudonyms.
Protocol~\ref{protocol1} formalizes the framework we propose.
At the first global training round, $A$ starts a federated learning task by initializing the global model $W^0$ and the global reputation vector $\gamma^0 \in \mathbb{R}^{K}$. The global reputation vector $\gamma$ is used to select participants for training and re-weight their mixed updates in the global model aggregation. 
Also, each participant $k$ initializes a local reputation vector $\zeta_k^0 \in \mathbb{R}^{K-1}$ that is used to store a local reputation value for each participant in the system different from $k$. 
In the first global training round, both $\gamma^0$ and $\zeta_k^0$ are initialized with zero vectors.
Later in this section, we explain how $\gamma$ and $\zeta_k$ are calculated and used in our framework.

\textbf{Adaptive selection of participants.}
At each round $t$, $A$ uses Procedure~\ref{proc1} to select a set $S$ of potential honest participants.
First, the first quartile of the global reputation $Q1_{\gamma^t}$ is computed.
Then, participants with a global reputation greater than or equal to $Q1_{\gamma^t}$ are assigned to the candidate set $Sc$.
Finally, the server selects a random set $S$ of $n=\max(C \cdot |Sc|, 2)$ participants from $Sc$.
As we explain later in this section, we expect the attackers 
(up to 20\%) to have global reputations less than $Q1_{\gamma^t}$. 
Note that, in the standard FL setting, $n = \max(C \cdot K, 1)$, but we replace $1$ by $2$ because FFL needs at least $2$ participants to be applicable. 
After that, $A$ sends the global model $W^t$ to the participants in $S$ through a secure channel. 
The server also adds the pseudonyms of the selected participants to a public board seen by all participants.

\textbf{Local training.}
Each participant $k\in S$ trains the received model on her private local data to obtain her local update $W_k$. 
Then, $k$ scales her computed local update by the number of data points she holds, $d_k$. 
The next step is for $k$ to randomly select another participant $j$ from the public board to exchange a fragment of her encrypted update $W_k$ with her.

\textbf{Participant selection for exchange.}
To avoid making herself a bridge for poisoning the global model, $k$ uses her local reputation vector $\zeta_k^t$ to decide whether to exchange fragments with any other participant $j$.
To do so, $k$ computes the first quartile of her local reputation vector $Q1_{\zeta_k^t}$ and assigns the participants in the public board that have local reputations greater than $Q1_{\zeta_k^t}$ to the set $Sc_k$.
Then, $k$ selects a random participant $j$ from $Sc_k$ and asks her to exchange fragments.
When $j$ receives the request from $k$ for fragment exchange, she checks the local reputation of $k$, 
$\zeta_{j, k}^t$, and if $\zeta_{j, k}^t<Q1_{\zeta_j^t}$, $j$ rejects $k$'s request for exchange.
Otherwise, $j$ accepts to exchange fragments with $k$, and they call protocol EXCHANGE\_FRAGMENTS. 
At the end of the protocol, both $k$ and $j$ obtain two encrypted mixed updates with their corresponding encrypted one-time pad (OTP) seeds: $({W_k})'_{mix}$ and $Enc_{{pk}_A}(s_{r_j})$ for participant $k$, and $({W_j})'_{mix}$ and $Enc_{{pk}_A}(s_{r_k})$ for participant $j$, and send them to $A$. Note that the seeds to generate the mixed updates' OTPs are encrypted under $A$'s public key and swapped between $k$ and $j$.
Later in this section, we detail how local reputations are computed.

\begin{algorithm}[!ht]
\footnotesize
\SetKwProg{Fn}{Function}{}{end}
\caption{Fragmented federated learning}
\label{protocol1}
\SetAlgoLined
\KwInput{$K, C, BS, E, \eta, T$}
\KwOutput{$W^T$, the global model after $T$ training rounds}

$A$ initializes $W^{0}, \gamma^0 = \{\gamma_k^0 = 0\}_{k = 1}^{K}$;

\For{each participant $k \in [1, K]$}
{
    $k$ initializes $\zeta_k^0 = \{\zeta_{k, j \ne k}^0 = 0\}$;
}

\For{each round $t \in [0, T-1]$}
{
     $S \leftarrow$\FuncSty{SELECT\_PARTICIPANTS($C, \gamma^t)$};

    $A$ sends $W^t$ to all participants in $S$
    
    \For{each participant $k \in S$ \textbf{in parallel }} 
        {
            $k$ calls $(W_{k}^{t+1})'_{mix}, Enc_{{pk}_A}(s_{r_j}) \leftarrow$
            \FuncSty{PARTICIPANT\_UPDATE($k, W^{t})$};
            
            $k$ sends $(W_{k}^{t+1})'_{mix}, Enc_{{pk}_A}(s_{r_j})$ to $A$;
    
            $A$ decrypts $Enc_{{pk}_A}(s_{r_j})$ with its private key to obtain $s_{r_j}$;
            
            $A$ generates $r_j \leftarrow PRNG(s_{r_j})$;
             
            $A$ decrypts $(W_{k}^{t+1})'_{mix}$ with $r_j$ as
            $(W_{k}^{t+1})_{mix} \leftarrow
            (W_{k}^{t+1})'_{mix} \oplus r_j$;
        }
        
    \FuncSty{COMPUTE\_SIMILARITY()};
    
    \FuncSty{UPDATE\_REPUTATIONS()};
        
    Let $\nu^{t}$ be the computed trust vector from Expression~(\ref{normalize_gr});
        
    $A$ aggregates 
    $W^{t} \leftarrow \frac{1}{\sum_{k \in S} d_k} \sum \nu_k^{t} (W_{k}^{t+1})_{mix}$\label{lin18};
    
}

\Fn{\FuncSty{PARTICIPANT\_UPDATE}($k, W^{t}$)}{
    {
        
        $W_k \leftarrow W^{t}$;  
     
        \For{each local epoch $e \in [1, E]$}
        {
	    \For {each batch $\beta$ of size $BS$}
            {
                $W_k \leftarrow W_k-\eta\nabla \mathcal{L}(W_k, \beta)$; 
            }
        }
        
         $W_k\leftarrow d_k W_k$;
         
         exchanged=false;
         
         Let $Q1_{\zeta_k^t}$ be the first quartile in $\zeta_k^t$;
         
         Let $Sc_k$ be the set of participants with local reputations greater than $Q1_{\zeta_k^t}$;
        
         \While{not exchanged}
         {
         $j \leftarrow$ select a random participant of $Sc_k$;
         
         \If {$\zeta_{j, k}^t \geq Q1_{\zeta_j^t}$} {
         
            $(W_{mix})', Enc_{{pk}_A}(s_{r_j})\leftarrow\FuncSty{EXCHANGE\_FRAGMENTS}(k,j)$;
         
         exchanged=true; 
         
            \Return  $(W_{mix})', Enc_{{pk}_A}(s_{r_j})$;
         }
       }
    }    
}
\end{algorithm}

\setcounter{algocf}{0}
\renewcommand*{\algorithmcfname}{Procedure}
\begin{algorithm}[!ht]
\footnotesize
  \DontPrintSemicolon
  \SetKwFunction{FMain}{SELECT_PARTICIPANTS}
  \caption{Adaptive selection of participants}
  \label{proc1}
  \SetKwProg{Fn}{}{:}{}
  \Fn{\FMain{$C, \gamma^t$}}{
  
            $Sc\leftarrow [ \ ]$; 
            
            Let $Q1_{\gamma^t}$ be the first quartile of $\gamma^t$;

            \For{each global reputation $\gamma_k^t \in \gamma^t$}
            {
            \If{$\gamma_k^t \geq Q1_{\gamma^t}$}{
                Add($k, Sc$);
            }
            }
            
        $n \leftarrow \max(C \cdot |Sc|, 2)$;
        
        $S\leftarrow$ select a random set of $n$ participants from $Sc$;

    \KwRet $S$;
}
\end{algorithm}

\textbf{Fragments exchanging and mixing.}
The protocol for exchanging fragments is key in our approach. 
We design this protocol to privately and efficiently exchange fragments of updates between participants. Specifically, a participant $k$ (the \emph{Initiator}) seeks to exchange a random fragment of her update with another participant $j$ (the \emph{Acceptor}). 
Both participants leverage a key exchange protocol to jointly generate two
complementing masks used to fragment their updates. The fragments are then
encrypted using one-time pads,
that is, random sequences added to the fragments, that allow {\em Initiator}
and {\em Acceptor} to combine
their updates but prevent each other from accessing their counterpart's update
parameters.
These one-time pads can only be removed by the central server after they
have been mixed
using the secret information provided by both participants.

Let $PRNG(\cdot)$ be a public pseudo-random number generator that takes an
integer as a seed.
Let $\mathbb{G}_p$ be the multiplicative group of integers modulo a large prime $p$ (2048-bit long), and $g$ a
generator of the group. 
All subsequent operations with vectors are performed coordinate-wise. Protocol EXCHANGE\_FRAGMENTS is as follows: 
\begin{enumerate}
    \item Participant $k$ randomly generates integers $a$, $s_{r_k}$ and $s_{\rho_k}$, 
    and sends $g^a \mod p$ and $Enc_{pk_{A}}(s_{r_k})$ to $j$. 
    
    \item Participant $j$ generates integers $b$, $s_{r_j}$ and $s_{\rho_j}$ and computes $g^b \mod p$, $r_j = PRNG(s_{r_j})$ and 
    $\rho_j = PRNG(s_{\rho_j})$, with the last two numbers 
    having bitlength $|W|\lambda$ (the bitlength of an update, that is, number of parameters times the bitlength of a parameter).  $j$ generates a mask $m=PRNG(g^{ab} \mod p)$ of 
    $0$'s and $1$'s of bitlength $|W|$ (equal to the number of update parameters) and its
    1-complement inverse mask $\lnot{m}$ such that $m \oplus \lnot{m} = \mathbf{1}_{|W|}$ (adding modulo 2 a mask to its 1-complement inverse yields the all ones mask). Then, $j$ sends $g^b \mod p$, $Enc_{pk_{A}}(s_{r_j})$, 
    $W_j \oplus r_j \oplus \rho_j$, 
    and $(W_j \odot \lnot{m}) \oplus \rho_j$ to $k$, where $\odot$ is an operator between an update and a mask that preserves the $i$-th update parameter if the $i$-th mask bit is 1, and clears to 0 the $i$-th update parameter if the $i$-th mask bit is 0.\\
 The random uniform binary mask $m$ will be used to fragment both $j$'s and $k$'s updates. 
    Also, both $r_j$ and $\rho_j$ are OTPs known to $j$ only,
    and used by $j$ to hide her original fragments from $k$. 
    \item $k$ computes $r_k = PRNG(s_{r_k})$ and $\rho_k = PRNG(s_{\rho_k})$ both of length $|W|\lambda$. $k$ generates $m=PRNG(g^{ab} \mod p)$ 
    and its inverse $\lnot{m}$. Then, $k$ computes her encrypted mixed update as follows:
    \begin{align}
    \label{exp1}
(W_k)^{'}_{mix} &= (W_k)_{mix} \oplus {r_j} \nonumber\\
                        &= W_k \oplus (W_k \odot m) \nonumber\\
                        &\oplus (W_j \oplus r_j \oplus \rho_j) \oplus ((W_j \odot \lnot{m}) \oplus \rho_j).
    \end{align}
    In Expression (\ref{exp1}), the result of subexpression $W_k \oplus (W_k \odot m)$ is to clear to 0 all parameters of $W_k$ at coordinates where the mask $m$ has a 1. On the other hand, subexpression $W_j \oplus (W_j \odot \lnot{m})$ clears to 0 all parameters of $W_j$ where the mask $m$ has a 0. Bitwise adding both subexpressions yields the mixed update $(W_k)_{mix}$. Note that the two appearances of $\rho_j$ cancel each other and $r_j$ encrypts the mixed update into $(W_k)^{'}_{mix}$.
    Since the mask $m$ is random, the mixed update can be expected to contain the same number of 1s and 0s. Hence, $(W_k)_{mix}$ can be expected to contain half of the coordinates from $W_j$
    and the other half from $W_k$. Another important remark is that the mixed update is encrypted using $r_j$, whereas $k$ has only $Enc_{pk_{A}}(s_{r_j})$, so $k$ cannot obtain the full cleartext mixed update. \\
    
    After that, $k$ sends $(W_k \oplus r_k \oplus \rho_k)$, and $(W_k \odot \lnot{m}) \oplus \rho_k$ to $j$. She
    also sends $(W_k)^{'}_{mix}, Enc_{pk_{A}}(s_{r_j})$ to the server.\\
    
Receiving $Enc_{pk_{A}}(s_{r_j})$ 
allows the server to decrypt the encrypted mixed update but does not allow it to extract any original fragment separately or link it to the fragment's originator, as we show in Appendix~\ref{sec:privacy}. Moreover, the use of the Diffie-Hellman key exchange method~\cite{diffie1976new} to generate seed  $g^{ab}$ ensures that no one but $k$ and $j$ knows the generated mask $m$. 
    
    \item Similarly, $j$ computes her encrypted mixed update as: 
    \begin{equation}
    \begin{aligned}
    (W_j)^{'}_{mix} &= (W_j)_{mix} \oplus {r_k} \\
                        &= W_j \oplus (W_j \odot m) \\
                        &\oplus (W_k \oplus r_k \oplus \rho_k) \oplus ((W_k \odot \lnot{m}) \oplus \rho_k).
    \end{aligned}
     \end{equation}
    Finally, $j$ sends  $(W_j)^{'}_{mix}, Enc_{pk_{A}}(s_{r_k})$ to the server.
\end{enumerate}
A naive and simpler way to exchange fragments would be to encrypt the updates coordinate by coordinate by using the server public key before exchanging fragments. 
However, this would cause significant communication and computation overheads for both the server and the par\-ti\-ci\-pants. Instead, our protocol uses OTPs as a means to encrypt the mixed updates. In other words, we rely on symmetric-key encryption, which is much more efficient when dealing with models that may contain millions of parameters. 

\textbf{Decryption of encrypted mixed updates.}
When at training round $t$ the server receives a mixed encrypted update $(W_{k}^{t+1})'_{mix}$ and its corresponding encrypted seed $Enc_{pk_{A}}(s_{r_j})$, it uses its private key to obtain the clear seed $s_{r_j} = Dec_{{sk}_{A}}(Enc_{pk_{A}}(s_{r_j}))$. 
Then $A$ regenerates $r_j$ using $s_{r_j}$ as a seed to $PRNG(\cdot)$
and bitwise adds $r_j$ to $(W_{k}^{t+1})'_{mix}$ to obtain $(W_{k}^{t+1})_{mix}$. 
$A$ does the same for all $k \in S$ to get the plain mixed updates $\{(W_{k}^{t+1})_{mix}|k\in S\}$. 

\textbf{Computation of reputation and trust values.} 
We explain how reputations are computed and used to neutralize potential attackers in the system.
First, the similarity between the mixed updates and their centroid are computed to measure their quality. Second, the global and local reputations are updated based on the computed similarities. Third, 
the global reputations are used to compute the trust values for the senders of the mixed updates, and these trust values finally weight the mixed updates when aggregating them to obtain 
the new global model. The above three steps
are detailed next:

{\em 1. Compute similarity.} 
In this procedure, the server computes the corresponding mixed gradients of the set $\{(W_{k}^{t+1})_{mix}|k\in S\}$. 
For a mixed update $(W_{k}^{t+1})_{mix}$, $A$ computes its corresponding gradient as
\begin{equation}
    \label{extract_gradient}
     (\nabla_{k}^{t})_{mix} = (W^t - (W_{k}^{t+1})_{mix})/\eta.
\end{equation}
The impact of untargeted attacks is expected to be evident in the magnitudes of whole poisoned updates, whereas the impact of the targeted attacks is expected to be more evident in the last-layer gradients that carry the indicative features and map directly to the prediction probability~\cite{fung2020limitations}.
Therefore, we independently analyze the whole mixed gradients and the last layer of mixed gradients and then combine the results to simultaneously counter both untargeted and targeted attacks. 
To do so, $A$ first computes the magnitude of each mixed gradient 
and obtains $\{||(\nabla_{k}^{t})_{mix}|||k\in S\}$. 
Then, it computes the median of the computed magnitudes, $med_{mix}^{t}$. 
Since most mixed updates are good, $med_{mix}^{t}$ is expected to fall amid good mixed updates.
After that, $A$ computes the distance between $med_{mix}^{t}$ and each magnitude $||(\nabla_{k}^{t})_{mix}||$ as
\begin{equation}
    \label{dist_sim}
    ds_k^{t} = |med_{mix}^{t} - ||(\nabla_{k}^{t})_{mix}|||.
\end{equation}
For each participant $k$, $ds_k^{t}$ represents how far the magnitude of the participant's mixed update is from $med_{mix}^{t}$. 
In the case of untargeted attacks, poisoned mixed updates are expected to deviate more from $med_{mix}^{t}$ and thus result in 
larger $ds^{t}$ values. 
It would also be possible to compute distances
based on the updates themselves and the median update rather
than their magnitudes, but this would 
cause a higher computational overhead on the server.
To capture the behavior of targeted attacks, $A$ extracts the mixed gradients of the last layer $L$ and obtains the set $\{(\nabla_{k}^{t})_{mix, L}|k\in S\}$. 
Then, $A$ computes $med_{mix, L}^{t}$ as the coordinate-wise 
median of the previous
set. Since honest participants share the same objective and are a majority, the median of the mixed last-layer gradients is expected to lie in the same direction as the honest participants' last-layer gradients. 
Thus, $A$ computes the cosine similarity between $med_{mix, L}^{t}$ and each $(\nabla_{k}^{t})_{mix, L}$ as
\begin{equation}
    \label{cos_sim}
    cs_k^{t} = \cos \varphi = \frac{(\nabla_{k}^{t})_{mix, L} \cdot med_{mix, L}^{t}}{||(\nabla_{k}^{t})_{mix, L}|| \cdot || med_{mix, L}^{t}||}.
\end{equation}
This yields a cosine similarity value $cs_k^{t+1} \in [-1, 1]$ for each participant's mixed update. 
Note that poisoned mixed updates, being a minority,
can be expected to have lower cosine similarity with $med_{mix, L}^{t}$ than good mixed updates. 

To compute a combined similarity value that captures the behaviors of both untargeted and targeted attacks, $A$ performs the following steps:
\begin{enumerate}
\item Normalize and invert the computed distances into the range $[0, 1]$ as
 $ds_k^{t} = 1 - ds_k^{t}/\max_{j\in S}(ds_j^{t}).$
\item Normalize the cosine similarity vector $cs^{t}$ into 
the range $[0, 1]$ as
    $cs_k^{t} = (cs_k^{t} + 1)/2.$
\item Compute the aggregated similarity value of $k$ as
        $sim_k^{t} = \alpha ds_k^{t} + (1 - \alpha) cs_k^{t}.$
\end{enumerate}
In this way, smaller $sim_k^{t}$ values are likely to correspond to potential poisoned mixed updates.
The hyperparameter $\alpha \in [0,1]$ is used to tune the simultaneous detection of untargeted and targeted attacks. 

{\em 2. Update reputations.} After computing the vector $sim^{t}$
containing similarities for all participants $k\in S$, $A$ updates the global reputations of those participants as
\begin{equation}
    \label{update_gr}
    \gamma_k^{t+1} = \gamma_k^{t} +  (sim_k^{t} - Q1_{sim^{t}}),
\end{equation}
where $Q1_{sim^t}$ is the first quartile of similarity values.
Based on its computed similarity $sim_k^{t}$, $k$'s global reputation increases or decreases: a similarity less than  $Q1_{sim^t}$
causes a reputation decrease and may lead to negative reputation.
If $k$ is an attacker and always sends poisoned mixed updates, its global reputation will get smaller and smaller and thus he will be excluded by the server from future selection for training.
However, an honest participant $k$ could also experience
a reputation decrease if she sent an update with a poisoned fragment during the exchange.
Therefore, $A$ also sends $sim_k^{t}- Q1_{sim^{t}}$ to $k$, who uses it to update the local reputation of the participant $j$ with whom $k$ has exchanged fragments in training round $t$:
\begin{equation}
    \label{update_lr}
    \zeta_{k, j}^{t+1} = \zeta_{k, j}^{t} + (sim_k^{t} - Q1_{sim^{t}}).
\end{equation}
This ensures that, when $k$ obtains a small similarity value because of $j$'s poisoned fragment, she reduces $j$'s local reputation.
As a result, if $j$ exchanges poisoned fragments with the other participants, she will get a bad local reputation among all honest participants and thus become an outcast, so that no honest participant will accept to exchange fragments with him in the future.

{\em 3. Adaptive model aggregation.} The server adaptively aggregates the received mixed updates using the global reputations of their senders. 
First, $A$ computes the trust vector $\nu^{t}$ as
\begin{equation}
    \label{normalize_gr}
    \nu^{t} = \max(\tanh{(\gamma^{t+1} - Q1_{\gamma^{t+1}})},0).
\end{equation}
The hyperbolic tangent function ($\tanh$) squashes its negative inputs into the range $[-1, 0[$,  and its positive inputs into the range $]0, 1]$.
Since the attackers are expected to have values lower than $Q1_{\gamma^{t+1}}$, $\tanh$ will make their trust scores in $\nu^{t}$ less than $0$. The maximum in Expression (\ref{normalize_gr}) sets the attackers' trust values to 0.
Note that, since the reputations of honest participants are likely to increase in every training round, their trust scores will converge to $1$ as training evolves, due to the use of $\tanh$.
Finally, the server uses the values in the computed trust vector $\nu^{t}$ to re-weight and aggregate the mixed updates 
(line~\ref{lin18} of Protocol~\ref{protocol1}). Since the 
attackers' trust values are 0, their updates are neutralized
in the aggregation.

\section{Privacy Analysis}
\label{sec:privacy}
In this section, we theoretically demonstrate the effectiveness of FFL against privacy attacks.

\subsection{Privacy between participants}

\begin{proposition}
\label{props0}
Two participants $j$ and $k$ exchanging fragments in Protocol EXCHANGE\_FRAGMENTS do not learn each other's fragments.
\end{proposition}

\begin{proof}
Since the protocol is symmetric, we only need to prove that participant $k$ does not learn participant $j$'s fragments. Participant $k$ receives the following from participant $j$:  
$g^b \mod p$, $Enc_{pk_{A}}(s_{r_j})$, 
    $W_j \oplus r_j \oplus \rho_j$, 
    and $(W_j \odot \lnot{m}) \oplus \rho_j$. 
Then participant $k$ can compute the common mask $m$, but she cannot decrypt $s_{r_j}$, which would allow her to re-create $r_j$. However, participant $k$ can add $W_j \oplus r_j \oplus \rho_j$ and $(W_j \odot \lnot{m}) \oplus \rho_j$, which, combined to $k$'s knowledge of $m$, allows $k$ to learn the bits of $r_j$ that encrypt parameters for which bits in the mask are 0. However, in the mixed update $(W_k)_{mix}$ computed by participant $k$ in Expression (\ref{exp1}), all parameters from participant $j$ corresponding to mask positions equal to 0 are cleared to 0. Hence, the bits of $r_j$ discovered by $k$ do not allow her to retrieve any parameter of $j$. 
\end{proof}

Note that a third-party intruder observing 
the exchange of fragments cannot do better than any of the two participants at learning the other participant's parameters. In fact, the intruder is likely to be in a worse position, because he does not know $m$ unless it is leaked by one of the participants. 

\subsection{Unlinking participants from their updates}
\label{unlik}

\begin{proposition}
\label{props1}
Given a mixed update $(W_k)_{mix}$ obtained by a participant $k$ after exchanging fragments, 
the probability that a certain subset of $u$ parameters in $W_k$ 
is entirely present in participant $k$'s mixed update $(W_k)_{mix}$ is $(1/2)^u$.
\end{proposition}

\begin{proof}
\label{proof1}
By construction of Protocol EXCHANGE\_FRAGMENTS,
in the mixed update $(W_k)_{mix}$ the original parameters of participant $k$ are found only where the mask $m$ has bits with value 0. Now, the probability a specific set of $u$ positions in $m$ being 0 is $(1/2)^u$ if the generator $PRNG$ used to obtain $m$ is good.
\end{proof}

A consequence of Proposition~\ref{props1} is that mixing updates effectively unlinks them from their originators: indeed, the probability
that a specific set of parameters in the original update survives in the mixed update decreases exponentially with the set size.
The effectiveness of unlinking updates against privacy attacks is examined in the next sections.

\subsection{Robustness against membership inference attacks}
\label{membership_attacks}

Membership inference attacks (MIAs) \cite{nasr2019comprehensive, melis2019exploiting} leverage a participant's local update $W_k$ to infer if a specific data point $(x, y)$ was part of her training data. 
In \cite{nasr2019comprehensive}, a semi-honest server exploits the distinguishable pattern that $(x, y)$ leaves on $W_k$. To carry out the attack, 
the server trains a binary classifier using some available data 
containing member and non-member data points and the components of a model trained on the member data points. The binary classifier predicts a membership score for any target data point $(x,y)$. The membership score is the probability that $(x,y)$ belongs to a target participant's training data. The attack components include: the calculated gradient vector on the data point, the activations of the intermediate layers of the participant's model, the activation of the output layer, and a scalar representing the loss of the model on the data point. The authors demonstrate that the gradient vector on the target data point is the most important component for the success of the attack. According to the way a semi-honest server performs MIAs, we can state the following proposition.

\begin{proposition}
\label{props2}
Given a mixed update $(W_k)_{mix}$ of a participant $k$, it is not possible to correctly predict the membership score of a target data point $(x,y)$ belonging to $k$ by using $(W_k)_{mix}$. 
\end{proposition}

\begin{proof}
\label{proof2}
The binary classifier needs the learned pattern (that is, the correct attack components) to correctly predict the membership score of $(x, y)$.
However, in FFL, the semi-honest server will compute an intermediate layer activation $f_{k, l}(x) = \sigma_l((w_{k, l})_{mix} \cdot x + (b_{k, l})_{mix})$ instead of $f_{k, l}(x) = \sigma_l(w_{k, l} \cdot x + b_{k, l})$. This will result in random activations and a random loss scalar as well. Based on that, calculating the gradient vector by backpropagating from a wrong loss scalar through 
the parameters of a mixed model will result in a completely random attack component, and hence in a random membership score prediction for the data point $(x, y)$. Therefore, FFL prevents the semi-honest server from correctly predicting the membership score of any target data point.
\end{proof}

\subsection{Robustness against property inference attacks}
\label{property_attacks}

Property inference attacks \cite{ganju2018property, melis2019exploiting} try to infer specific properties about the training data by recognizing patterns within a participant's local model. \cite{melis2019exploiting} show how to infer properties of a participant's training data that are uncorrelated to the main task features. The idea of the attack is to use a participant's update to infer properties that characterize a subset of her training data. A semi-honest server can use some auxiliary data, with and without the property, to generate updates with and without that property. Then, it uses the gradients of the generated updates as input features to train a binary classifier. The binary classifier is used to distinguish if an input gradient was computed on data with the same target property. Finally, when the server receives a target participant's update $W_k$, it extracts its gradient as 
\begin{equation}
    \label{eq_compute_grad}
    \nabla W_k = \frac{W - W_k}{\eta}.
\end{equation}
Then the server passes $\nabla W_k$ to the binary classifier to determine if the participant's data have the target property.

We can notice that, like the classifier in the MIA attack, the classifier in this attack  mainly depends on the original gradient vectors.
FFL prevents the server from obtaining 
the whole original gradient of the target participant. In general, her mixed gradient is inconsistent with the pattern the binary classifier was trained on. Thus, the classifier decision will most likely be inaccurate. 

\subsection{Robustness against reconstruction attacks}
\label{reconstruction_attacks}

Reconstruction attacks \cite{zhu2020deep, zhao2020idlg, geiping2020inverting} are much stronger than the previous ones, since they can extract both the original training inputs and the labels from a participant's local gradient. The idea behind these attacks is that a participant's update $W_k$ is computed based on both the global model $W$ and the participant's training data $(x_k, y_k)$. 
Since a semi-honest server has both $W$ and $W_k$, it can obtain the training data by inverting the update gradient. First, the server computes participant $k$'s gradient $\nabla W_k$ by using Expression \eqref{eq_compute_grad}. After that, it tries to invert the gradient and to find the unknown training data $(x_k, y_k)$ that result in the same extracted gradient: it starts by randomly initializing a dummy input $x^*$ and a label input $y^*$, and feeds these ``dummy data'' to the global model $W$ to get ``dummy gradients'' $\nabla W^*$ as:
\begin{equation}
    \label{eq_ra_loss}
    \nabla W^* = \nabla \mathcal{L}(W, (x^*, y^*))
\end{equation}
Then, the server repeatedly modifies the dummy data in an adversarial perturbation way, based on the difference between the dummy gradient $\nabla W^*$ and participant $k$'s original gradient $\nabla W_k$. A small distance between $\nabla W^*$ and $\nabla W_k$ means that the dummy data are similar to the original data. The authors of \cite{zhu2020deep, zhao2020idlg} use the Euclidean distance between $\nabla W_k$ and $\nabla W^*$ as the objective function to modify the dummy data, whereas \cite{geiping2020inverting} employ the cosine similarity. 
The objective function used in \cite{geiping2020inverting} is given by
\begin{equation}
    \label{eq_rag_loss}
    argmin_{x^*, y^*} \left(1 - \frac{<\nabla W_k, \nabla W^*>}{||\nabla W_k||||\nabla W^*||}\right)
\end{equation}

Based on the above we can state the following proposition.
\begin{proposition}
\label{props3}
A mixed gradient ${(\nabla W_k)}_{mixed}$, sent by a participant $k$, cannot be leveraged by the server to estimate a target data point $(x,y)$ in $k$'s local data. 
\end{proposition}
\begin{proof}
\label{proof3}
 When using a mixed gradient ${(\nabla W_k)}_{mixed}$ instead of the original $\nabla W_k$, the objective functions used in the
 reconstruction attacks will in general result in a random reconstructed data point because in general there is no original data point corresponding to that mixed gradient.
\end{proof}

Proposition~\ref{props3} guarantees that by providing the server with mixed updates instead of the original ones, FFL effectively prevents the semi-honest server from performing reconstruction attacks.

\section{Security Analysis}
\label{sec:security}
When exchanging fragments with another honest participant $k$, an attacker $j'$ can follow one of three strategies:
\begin{itemize}
    \item \textbf{Strategy 1}. Exchange her poisoned fragment with $k$ 
and send a poisoned mixed update to the server containing her other poisoned fragment and the good fragment of $k$. With an expected number of $n/5$ attackers in a training round, there is a chance of poisoning at most $2n' = 2n/5$ mixed updates and at most $n' = n/5$ coordinates.
    
    \item \textbf{Strategy 2}. Exchange her poisoned fragment with $k$ 
and send a fully poisoned mixed update to the server. 
    Thus, there is a chance of poisoning at most $2n' = 2n/5$ mixed updates and at most $2n' = 2n/5$ coordinates.

    \item \textbf{Strategy 3}. 
Exchange her poisoned fragments with the other participants and 
submit fully good updates to the server.
The attacker does this to increase his global reputation and 
discredit the honest participants in front of the server.
This strategy poisons fewer updates than Strategies 1 or 2: 
at most $n' = n/5$ mixed updates and at most $n' = n/5$ coordinates. 
\end{itemize}

Now let us see how FFL can counter the above strategies and neutralize the impact of poisoned mixed updates on the global model aggregation. 
In \emph{Strategy 1}, since the number of untouched good coordinates is a majority ($4n/5$), the poisoned mixed updates will have less similarity to the centroid of the mixed update. 
This lower similarity will decrease the global reputations of both attackers and some honest participants, and the local reputations of the attackers.
But since an honest participant is more likely to select another honest participant for the exchange, honest participants will find an opportunity to increase their global reputations and offset the harm caused by the attackers' poisoned fragments.
That is because the probability of an honest participant selecting another honest participant is $(n - n' - 1)/(n-1) = (4n - 5)/(5n - 5)$, whereas the probability of selecting an attacker is $(n')/(n-1) = n/(5n - 5)$. 
Moreover, as the training evolves, the attackers will obtain smaller and smaller local reputations (less than the first quartile in the local reputation vectors). As a result, honest participants will not accept exchanging fragments with them. 
This will force attackers to send their poisoned updates directly 
to the server or keep them. 
If they send them, their global reputations will decrease more and more to be below the first quartile $Q1_{\gamma}$, which will completely neutralize their influence on the global model, because they will not be selected for future training. 
On the other hand, if they refrain from sending their poisoned updates, they will neutralize themselves. 
Note that, even if some attackers managed to have global reputations slightly greater than the first quartile in the early training rounds, they would have less influence on the global model aggregation than honest participants, because they will have small trust scores.
Fig.~\ref{trust_evol} shows an example of how the average trusts of honest participants and attackers evolve as the training evolves when the attackers 
follow Strategy 1.
\begin{figure}[!ht]
    \centering
      \includegraphics[width=0.8\linewidth]{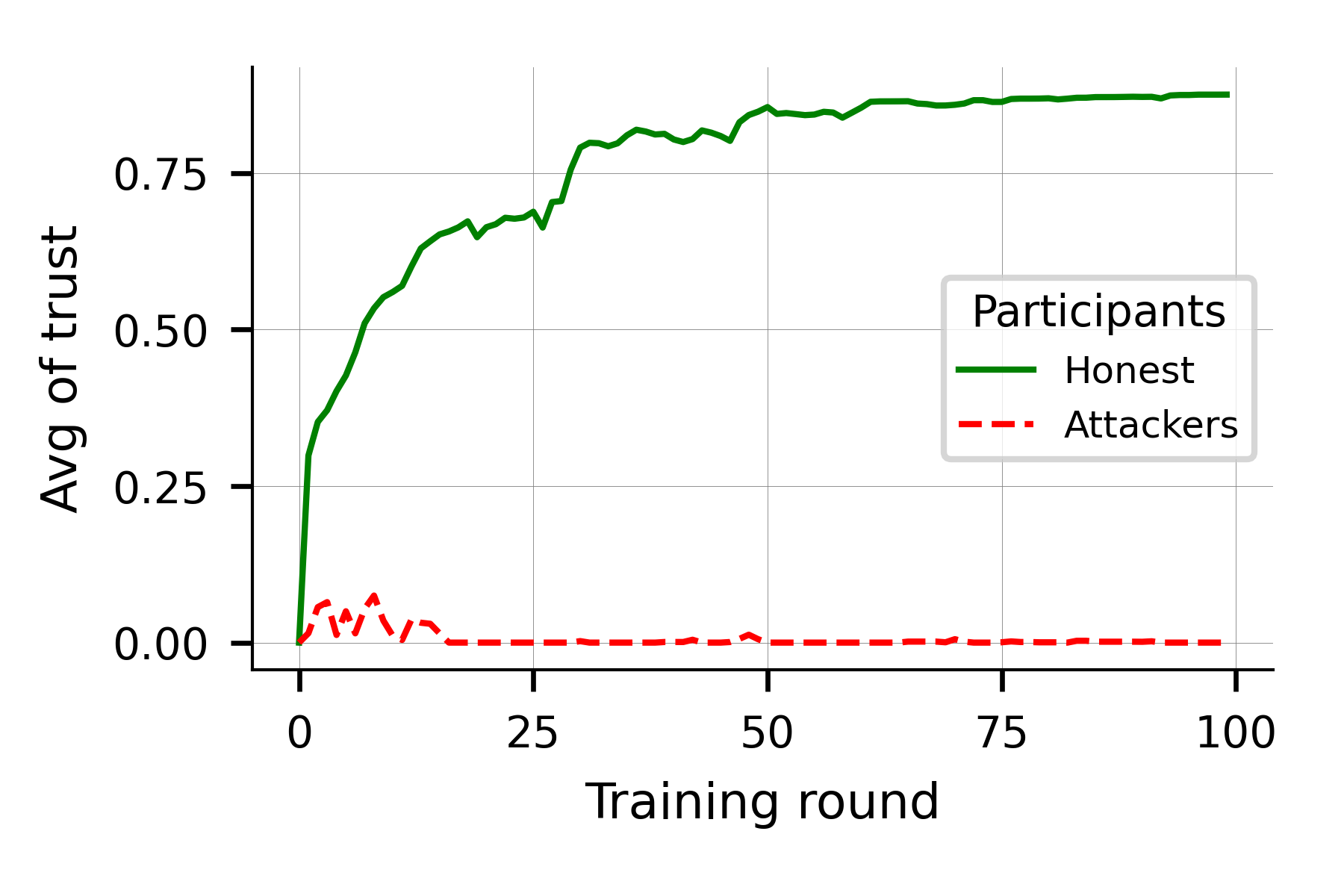}
      \caption{Participant average trust evolution during training in the CIFAR10-VGG16 benchmark.}
      \label{trust_evol}
\end{figure}
In \emph{Strategy 2}, the situation does not differ much. The attackers will get still lower similarity values than under Strategy 1, 
because they send fully poisoned updates,
 which will be farther from the centroid than the poisoned mixed updates sent 
by the honest participants. 
This will be reflected in their global and local reputations, which will lead to attackers getting neutralized.
As for \emph{Strategy 3}, the attackers will get good global reputations because they always send good mixed updates; however, they will get 
low local reputations with honest participants. 
This will deter honest participants from making future exchanges with the attackers. Hence, the attackers will end up being unable 
to poison the mixed updates sent by honest participants.

\section{Experimental Analysis}
\label{sec:setup}

In this section, we report empirical results on three real data sets for the most relevant security and privacy attacks discussed above.  
We used the PyTorch framework to implement the experiments on a computer with an AMD Ryzen 5 3600 6-core CPU, 32 GB RAM, an NVIDIA GTX 1660 GPU with 6 GB RAM, and Windows 10 OS. 
As said above, with PyTorch, the parameter bitlength is $\lambda=32$.  
Our code and data are available for reproducibility purposes\footnote{\url{https://github.com/anonymized30/FFL}}.

In all the experiments, the {\em Acceptor}, respectively the {\em Initiator}, looped through all the layers of her update and generated a random binary $mask_{l}$ for each layer $l\in [1, L]$, where $L$ is the total number of layers in the global model $W$. At the end, the {\em Acceptor}, resp. the {\em Initiator}, set
$mask=mask_{l}||\ldots||mask_L$, where $||$ is the concatenation operator, and used $mask$ to exchange a random fragment.

We used the Diffie-Hellman key exchange protocol with the secure 2048-bit MODP group~\cite{rfc5114} to generate the secret shared seeds of the fragments' masks, and we used~\cite{pycryptodome} to encrypt and decrypt the OTP seeds with the recommended 3072-bit key size.
We used a value of $\alpha = 0.2$ because we found it gives better simultaneous detection of untargeted and targeted attacks. 
That is because targeted attacks are stealthier than untargeted ones.

\subsection{Data sets and models}
We tested the proposed method on three ML tasks: tabular data classification, image classification and sentiment analysis. Table~\ref{tab:datasets_models} summarizes the data sets and models we used. 
To evaluate the effectiveness of FFL at defeating reconstruction attacks, we used the code provided by the authors of \cite{geiping2020inverting}, who perform reconstruction attacks with the \emph{ConvNet64} model described in their paper (with about 3 million parameters) and some images from the \emph{CIFAR10} validation set.

\begin{table}[ht]
\centering
\scriptsize
\caption{Data sets and models used in the experiments}
\label{tab:datasets_models}
\resizebox{0.5\textwidth}{!}{%
\begin{tabular}{|l|c|c|c|c|}
\hline
\multicolumn{1}{|c|}{Task}      & Data set & \# Examples & Model  & \# Parameters \\ \hline
Tabular classif.                & Adult    & 48,842      & MLP    & $\sim$5K            \\ \hline
\multirow{2}{*}{Image classif.} & MNIST    & 70K         & CNN    & $\sim$22K           \\ \cline{2-5} 
                                & CIFAR10  & 60K         & VGG16  & $\sim$15M           \\ \hline
Sent. analysis                  & IMDB     & 50K         & BiLSTM & $\sim$13M           \\ \hline
\end{tabular}%
}
\end{table}

\textbf{Tabular data classification.} We used the Adult tabular data set\footnote{\url{https://archive.ics.uci.edu/ml/datasets/adult}} that contains $48,842$ records of census income information with 14 numerical and categorical attributes. 
The class label is the attribute \emph{income} that classifies records into either $> 50$K or $\leq 50$K.
We used $80\%$ of the data as training data and the remaining $20\%$ as validation data.
We randomly and uniformly split the $80\%$ training examples among $20$ FL participants. 
We used a multi-layer perceptron (MLP) with one input layer, one hidden layer and one output layer that contains about $5$K learnable parameters.
The output layer is followed by a Sigmoid function to produce the final predicted class for an input record. 
The MLP was trained during $100$ rounds. In each round, the FL server selected $10$ participants and asked them to train the model for $1$ local epoch and a local batch size $64$. 
The participants used the binary cross-entropy with logit loss function and the Adam optimizer with a learning rate = $0.001$ to train their models.

\textbf{Image classification.} We used two data sets for this task:
\begin{itemize}
\item MNIST data set. It contains $70$K handwritten digit images from $0$ to $9$ ({\em i.e.}, $10$ classes) \cite{lecun1999object}. 
The images are in grayscale with size $28\times28$ pixels, and they are divided into a training set ($60$K examples) and a testing set ($10$K examples). 
We randomly and uniformly split the $60$K training examples among $100$ simulated participants of an FL setting. 
We used a two-layer convolutional neural network (CNN) with two fully connected layers. 
The CNN model was trained during $200$ rounds. In each round, the FL server randomly chose $50$ participants and asked them to train the model for $3$ local epochs and a local batch size $64$. 
The participants used the cross-entropy loss function and the stochastic gradient descent (SGD) optimizer with a learning rate = $0.001$ and momentum = $0.9$ to train their models.

\item CIFAR10 data set. It consists of $60$K colored images of $10$ different classes \cite{krizhevsky2009learning}. 
The data set is divided into $50$K training examples and $10$K testing examples. 
We randomly and uniformly split the $50$K training examples among $20$ FL participants.
We used the VGG16 CNN model with one fully connected layer \cite{simonyan2014very}. 
The VGG16 model was trained during $100$ rounds. 
In each round, the FL server randomly chose $10$ participants and asked them to train the model for $3$ local epochs and a local batch size $32$. 
The participants used the cross-entropy loss function and the SGD optimizer with a learning rate = $0.01$ and momentum = $0.9$ to train their models.
\end{itemize}

{\bf Sentiment analysis.} We used the IMDB Large Movie Review data set \cite{maas2011learning} for this binary sentiment classification task. 
This data set is a collection of $50$K movie reviews and their corresponding sentiment binary labels (either positive or negative). 
We divided the data set into $40$K training examples and $10$K testing examples. 
We randomly and uniformly split the $40$K training examples among $20$ FL participants. 
We used a Bidirectional Long/Short-Term Memory (BiLSTM) model, which has an embedding layer of $100$ dimensions for each token. 
The model ends with a linear layer followed by a Sigmoid function to produce the final predicted sentiment for an input review. 
The BiLSTM was trained during $100$ rounds. In each 
round, the FL server randomly chose $10$ participants and asked them to train the model for $1$ local epoch and a local batch size $32$. 
The participants used the binary cross-entropy with logit loss function and the Adam optimizer with a learning rate = $0.001$ to train their models.

\textbf{Evaluation metrics.}
\label{metrics} We used the following evaluation metrics on the test set examples to assess the impact of the attacks on the learned model and the performance of the proposed framework w.r.t. the state of the art:
\begin{itemize}
\item \emph{Test error (TE)}. This is the error resulting from the loss functions used in training. The lower the test error, the more robust the method is against the attack.
\item \emph{Overall accuracy (All-Acc)}. This is the number of correct predictions divided by the total number of predictions. 
\item \emph{Source class accuracy (Src-Acc)}. We evaluated the 
accuracy for the subset of test examples belonging to the source class. One may achieve a good overall accuracy while degrading the accuracy of the source class. 
\item \emph{Attack success rate (ASR)}. This is the proportion of targeted examples (with the source label) that are incorrectly classified into the label desired by the attacker.
\end{itemize}
An effective defense needs to retain the benign performance of the global model on the main task while reducing ASR. 

\subsection{Robustness against security attacks}

\label{exp:robustness}
We next report results on the robustness of FFL against two security
attack strategies: in {\em Strategy 1}, the attacker generates two poisoned fragments, one of which he exchanges with 
another participant $k$, and he sends to the server a mixed update
consisting of his other poisoned fragment and the good fragment of $k$; in {\em Strategy 2}, the attacker exchanges a poisoned fragment with $k$ and sends a poisoned update to the server. 

Also, we compare the performance of FFL with the performance of FedAvg\cite{mcmahan2017communication}, the median~\cite{yin2018byzantine}, the trimmed mean~\cite{yin2018byzantine}, multi-Krum~\cite{blanchard2017machine}, and FoolsGold~\cite{fung2020limitations} on standard FL. 
Notice that FedAvg does nothing to counter security attacks (it systematically aggregates all received updates).

We evaluated TE and All-Acc under the Gaussian noise untargeted attack, and TE, Src-Acc and ASR under the label-flipping targeted attack. 
We used standard FL with FedAvg when no attacks were performed as a baseline to show the impact of attacks on the model performance. 
In our experiments, the percentage of attackers was $20\%$ for all four benchmarks.

\textbf{Gaussian noise attack.} The attackers added Gaussian noise to their updates to prevent the model from converging \cite{li2019rsa, wu2020federated}. Specifically, they added noise with $0$ mean and $0.5$ standard deviation for the MLP and CNN model parameters, and $0.2$ standard deviation for VGG16 and BiLSTM parameters.
\begin{table*}[t]
\scriptsize
\centering
\caption{Robustness against Gaussian noise attacks. Best scores are in bold.}
\label{tab:robustness_gn}
\begin{tabular}{|cl|cc|cc|cc|cc|}
\hline
\multicolumn{2}{|c|}{\multirow{2}{*}{\textbf{\begin{tabular}[c]{@{}c@{}}Benchmark/ \\ Method\end{tabular}}}} & \multicolumn{2}{c|}{Adult-MLP}          & \multicolumn{2}{c|}{MNIST-CNN}         & \multicolumn{2}{c|}{CIFAR10-VGG16}      & \multicolumn{2}{c|}{IMDB-BiLSTM}       \\ \cline{3-10} 
\multicolumn{2}{|c|}{}                                                                                       & \multicolumn{1}{c|}{TE}     & All-Acc\% & \multicolumn{1}{c|}{TE}    & All-Acc\% & \multicolumn{1}{c|}{TE}     & All-Acc\% & \multicolumn{1}{c|}{TE}    & All-Acc\% \\ \hline
\multicolumn{1}{|c|}{\multirow{6}{*}{FL}}                        & FedAvg (no attacks)                       & \multicolumn{1}{c|}{0.349}  & 82.56     & \multicolumn{1}{c|}{0.112} & 96.69     & \multicolumn{1}{c|}{0.881}  & 80.77     & \multicolumn{1}{c|}{0.475} & 88.63     \\ \cline{2-10} 
\multicolumn{1}{|c|}{}                                           & FedAvg                                    & \multicolumn{1}{c|}{1.198}  & 75.08     & \multicolumn{1}{c|}{2.322} & 9.65      & \multicolumn{1}{c|}{10.324} & 10.0      & \multicolumn{1}{c|}{0.618} & 67.7      \\ \cline{2-10} 
\multicolumn{1}{|c|}{}                                           & Median                                    & \multicolumn{1}{c|}{0.349}  & \textbf{82.87}     & \multicolumn{1}{c|}{0.115} & 96.62     & \multicolumn{1}{c|}{1.020}  & 78.91     & \multicolumn{1}{c|}{0.524} & 88.33     \\ \cline{2-10} 
\multicolumn{1}{|c|}{}                                           & Trimmed mean                              & \multicolumn{1}{c|}{0.350}  & 82.61     & \multicolumn{1}{c|}{0.114} & 96.64     & \multicolumn{1}{c|}{1.046}  & \textbf{80.23}     & \multicolumn{1}{c|}{\textbf{0.457}} & \textbf{88.79}     \\ \cline{2-10} 
\multicolumn{1}{|c|}{}                                           & Multi-Krum                                & \multicolumn{1}{c|}{0.350}  & 82.69     & \multicolumn{1}{c|}{0.126} & 96.18     & \multicolumn{1}{c|}{0.998}  & 78.86     & \multicolumn{1}{c|}{0.565} & 87.72     \\ \cline{2-10} 
\multicolumn{1}{|c|}{}                                           & FoolsGold                                 & \multicolumn{1}{c|}{44.051} & 69.47     & \multicolumn{1}{c|}{2.741} & 8.92      & \multicolumn{1}{c|}{12.813} & 9.02      & \multicolumn{1}{c|}{2.694} & 51.61     \\ \hline
\multicolumn{1}{|c|}{\multirow{2}{*}{FFL}}                       & Strategy 1                                & \multicolumn{1}{c|}{\textbf{0.349}}  & 82.84     & \multicolumn{1}{c|}{\textbf{0.112}} & \textbf{96.67}     & \multicolumn{1}{c|}{\textbf{0.931}}  & 79.74     & \multicolumn{1}{c|}{0.536} & 88.15     \\ \cline{2-10} 
\multicolumn{1}{|c|}{}                                           & Strategy 2                                & \multicolumn{1}{c|}{0.350}  & 82.86     & \multicolumn{1}{c|}{0.113} & 96.61     & \multicolumn{1}{c|}{0.938}  & 79.26     & \multicolumn{1}{c|}{0.542} & 87.48     \\ \hline
\end{tabular}
\end{table*}

Table~\ref{tab:robustness_gn} shows the results under this attack.
First, we can see the significant negative impact of the attack on the performance of FedAvg regarding both the test error and the overall accuracy. 
The case was even worse with FoolsGold because the added noise made the attackers' last layers more diverse than the honest participants'.
FooldGold assumes participants with similar last layers to be attackers and those with diverse last layers to be honest.
Thus, it considered poisoned updates and excluded good updates in the model aggregation. 
The rest of the methods, including FFL, achieved comparable results to the baseline.
Since the added noise made the magnitudes of the poisoned updates different from those of good updates, we observe that i) the median and the trimmed mean were able to neutralize the poisoned parameters in model aggregation, ii) multi-Krum was able to exclude the poisoned updates due to the larger Euclidean distances they had, and iii) FFL was able to exclude poisoned mixed updates because their deviations from their medians were larger than those of good mixed updates.
We can also see that, in most cases (Adult-MLP, MNIST-CNN and CIFAR10-VGG16), FFL achieved the lowest test error among all methods.
As the training evolved, FFL set the trust values of honest participants to $1$ and thus fully considered their contributions. 
Note that FFL achieved similar performance under attack strategies 1 and 2.

\textbf{Label-flipping attack.} In the label-flipping attack, attackers flip the labels of correct training examples from one class (a.k.a. the source class) to another class and train their models according to the latter~\cite{biggio2012poisoning,fung2018mitigating}. 
For Adult, the attackers flipped each example with the label ''$>50$K'' to ''$\leq 50$K'', while they flipped each example with the label ''$7$'' to ''$1$'' for MNIST.
For CIFAR10, the attackers flipped each example with the label ''Cat'' to ''Dog''.
For IMBD, they flipped the ''positive'' reviews to ''negative''.

\begin{table*}[t]
\centering
\scriptsize
\caption{Robustness against label-flipping attacks. Best scores are in bold.}
\label{tab:robustness_lf}
\begin{tabular}{|cl|ccc|ccc|ccc|ccc|}
\hline
\multicolumn{2}{|c|}{\multirow{2}{*}{\textbf{\begin{tabular}[c]{@{}c@{}}Benchmark/ \\ Method\end{tabular}}}} & \multicolumn{3}{c|}{Adult-MLP}                                      & \multicolumn{3}{c|}{MNIST-CNN}                                      & \multicolumn{3}{c|}{CIFAR10-VGG16}                                  & \multicolumn{3}{c|}{IMDB-BiLSTM}                                    \\ \cline{3-14} 
\multicolumn{2}{|c|}{}                                                                                       & \multicolumn{1}{c|}{TE}    & \multicolumn{1}{c|}{Src-Acc\%} & ASR\% & \multicolumn{1}{c|}{TE}    & \multicolumn{1}{c|}{Src-Acc\%} & ASR\% & \multicolumn{1}{c|}{TE}    & \multicolumn{1}{c|}{Src-Acc\%} & ASR\% & \multicolumn{1}{c|}{TE}    & \multicolumn{1}{c|}{Src-Acc\%} & ASR\% \\ \hline
\multicolumn{1}{|c|}{\multirow{6}{*}{FL}}                        & FedAvg (no attacks)                       & \multicolumn{1}{c|}{0.350} & \multicolumn{1}{c|}{39.06}      & 60.94 & \multicolumn{1}{c|}{0.112} & \multicolumn{1}{c|}{95.43}     & 0.49  & \multicolumn{1}{c|}{0.881} & \multicolumn{1}{c|}{70.73}      & 14.51  & \multicolumn{1}{c|}{0.519} & \multicolumn{1}{c|}{87.16}     & 12.84 \\ \cline{2-14} 
\multicolumn{1}{|c|}{}                                           & FedAvg                                    & \multicolumn{1}{c|}{0.374} & \multicolumn{1}{c|}{10.98}      & 89.02 & \multicolumn{1}{c|}{0.130} & \multicolumn{1}{c|}{93.39}     & 1.26  & \multicolumn{1}{c|}{0.913} & \multicolumn{1}{c|}{53.82}      & 29.34  & \multicolumn{1}{c|}{0.562} & \multicolumn{1}{c|}{63.61}     & 36.39 \\ \cline{2-14} 
\multicolumn{1}{|c|}{}                                           & Median                                    & \multicolumn{1}{c|}{0.354} & \multicolumn{1}{c|}{31.28}      & 68.72  & \multicolumn{1}{c|}{0.117} & \multicolumn{1}{c|}{93.48}     & 1.07  & \multicolumn{1}{c|}{0.943} & \multicolumn{1}{c|}{58.01}      & 24.81  & \multicolumn{1}{c|}{0.643} & \multicolumn{1}{c|}{59.92}     & 40.08 \\ \cline{2-14} 
\multicolumn{1}{|c|}{}                                           & Trimmed mean                              & \multicolumn{1}{c|}{0.353} & \multicolumn{1}{c|}{32.12}     & 67.88 & \multicolumn{1}{c|}{0.118} & \multicolumn{1}{c|}{93.58}     & 1.07  & \multicolumn{1}{c|}{0.943} & \multicolumn{1}{c|}{58.43}      & 25.54  & \multicolumn{1}{c|}{0.649} & \multicolumn{1}{c|}{58.20}      & 41.80  \\ \cline{2-14} 
\multicolumn{1}{|c|}{}                                           & Multi-Krum                                & \multicolumn{1}{c|}{\textbf{0.350}} & \multicolumn{1}{c|}{39.34}     & 60.66 & \multicolumn{1}{c|}{0.113} & \multicolumn{1}{c|}{\textbf{95.33}}     & \textbf{0.39}  & \multicolumn{1}{c|}{0.896} & \multicolumn{1}{c|}{42.60}      & 39.33  & \multicolumn{1}{c|}{0.886} & \multicolumn{1}{c|}{39.05}     & 60.95 \\ \cline{2-14} 
\multicolumn{1}{|c|}{}                                           & FoolsGold                                 & \multicolumn{1}{c|}{0.351} & \multicolumn{1}{c|}{\textbf{39.72}}     & \textbf{60.28} & \multicolumn{1}{c|}{\textbf{0.111}} & \multicolumn{1}{c|}{95.14}     & \textbf{0.39}  & \multicolumn{1}{c|}{0.989} & \multicolumn{1}{c|}{67.41}      & 16.62  & \multicolumn{1}{c|}{0.549} & \multicolumn{1}{c|}{\textbf{86.86}}     & \textbf{13.14} \\ \hline
\multicolumn{1}{|c|}{\multirow{2}{*}{FFL}}                       & Strategy 1                                & \multicolumn{1}{c|}{\textbf{0.350}} & \multicolumn{1}{c|}{39.66}     & 60.34 & \multicolumn{1}{c|}{\textbf{0.111}} & \multicolumn{1}{c|}{\textbf{95.33}}     & \textbf{0.39 } & \multicolumn{1}{c|}{\textbf{0.849}} & \multicolumn{1}{c|}{\textbf{68.90}}      & \textbf{13.20}  & \multicolumn{1}{c|}{0.544} & \multicolumn{1}{c|}{86.64}     & 13.36 \\ \cline{2-14} 
\multicolumn{1}{|c|}{}                                           & Strategy 2                                & \multicolumn{1}{c|}{\textbf{0.350}} & \multicolumn{1}{c|}{39.50}      & 60.50  & \multicolumn{1}{c|}{0.112} & \multicolumn{1}{c|}{95.04}     & 0.49  & \multicolumn{1}{c|}{0.892} & \multicolumn{1}{c|}{68.80}      & 15.00  & \multicolumn{1}{c|}{\textbf{0.524}} & \multicolumn{1}{c|}{86.16}     & 13.84 \\ \hline
\end{tabular}
\end{table*}

Table~\ref{tab:robustness_lf} shows the results under the label-flipping attack.
For the Adult-MLP benchmark, the FedAvg performance under the attack significantly degraded for Src-Acc and ASR.
However, TE slightly increased and kept close to that of the baseline (FedAVg - no attacks). The reason for that is the Adult data set is highly imbalanced, with a skew toward the ''$\leq 50$K'' class label. The median and the trimmed mean achieved low TE, but saw degraded performance for Src-Acc and ASR because they discarded a large number of coordinates in the global model aggregation. Multi-Krum, FoolsGold and FFL achieved similar results to the baseline with slightly greater Src-ACC and slightly lower ASR.
Since the the Adult data set's training data were randomly and uniformly distributed among the participants, some of them had larger percentages of the target class examples compared to the original class distribution of the data set. The original percentage of the examples belonging to the target class ''$\leq 50$K'' in the data set is about $75.22\%$, which caused bias in the global model against the minority class ''$>50$K''. This made the updates of those honest participants having higher percentages of the target class close to the attackers' updates.
Since multi-Krum, FoolsGold and FFL excluded or penalized some honest participants with updates close to the attackers' updates, the global model became less biased against the minority class ''$>50$K'' and, hence, the source class accuracy slightly increased.
For MNIST-CNN, the performance of FedAvg, median and trimmed mean slightly degraded compared to the baseline.
An interesting note is that FedAvg was not highly affected by the attack. The reasons for that are the small size of the model and the simple and balanced distribution of the data set. That was also observed in~\cite{shejwalkar2021back}, where the authors argued that, in some cases, FedAvg could be robust against poisoning attacks.
For CIFAR-VGG16, Src-Acc degraded from $70.73\%$ to $53.82\%$, and ASR increased from $14.51\%$ to $29.34\%$ with FedAvg.
On the other hand, FedAvg achieved TE lower than that of the median, the trimmed median and FoolsGold.
That is because the attackers flipped the labels in the examples for only one class and kept the labels of the other classes unchanged.
The median and the trimmed mean decreased Src-Acc and increased ASR because of the large size of the VGG16 model. \cite{chang2019cronus} have shown that the estimation errors of the median and the trimmed mean scale up with the size of the model in a square-root manner.
Multi-Krum achieved the worst performance regarding Src-Acc and ASR because the small impact of the attack was not detectable in the large model.
Therefore, multi-Krum identified some attackers as honest while identifying 
some honest participants as attackers, which led to its poor performance.
FFL achieved the best performance among all the methods for all three metrics. 
FoolsGold scored the second-best after FFL regarding Src-Acc and ASR.
FFL and FoolsGold achieved such good performance because they analyzed the 
last-layer gradients, which contain more useful information for detecting the behavior of targeted poisoning attacks.
However, FoolsGold achieved the greatest TE among the other methods because it did not consider the full contributions of the honest participants.
FFL, however, considered almost all the full contributions of the honest participants and thus achieved the lowest TE. 

For the IMDB-BiLSTM benchmark, FFL and FoolsGold achieved the 
best performance among all methods, most of which were negatively 
impacted by the large size of the BiLSTM model.
FFL and FoolsGold outperformed the other methods by a large margin in
achieving good values for all metrics.
FoolsGold performed well in this benchmark because it was its ideal setting: updates from honest participants were somewhat different due to the different reviews they gave, whereas updates for attackers became very close to each other because they shared the same objective. 

To summarize, the results show that FFL can effectively defend against untargeted and targeted poisoning attacks while preserving the benign model performance. Moreover, FFL outperforms the state-of-the-art defenses in achieving good 
model performance while preventing the attackers from mounting successful security attacks and hindering the semi-honest server from mounting privacy attacks. 

\subsection{Protection against the reconstruction privacy attack}
\label{exp:rec_attack}

We next report results on the protection offered by FFL against the most powerful privacy attack, namely the reconstruction attack proposed in \cite{geiping2020inverting}. Notice that this attack does not require auxiliary data and 
can estimate the private training data by inverting their corresponding gradients. 

Fig.~\ref{fig:one_image_construction} shows the results when two participants, $k$ and $j$, sent their updates computed on 
just their private images (left column of the figure) to the server. In the FL setting (middle column of the figure), the server was able to reconstruct their private images with high accuracy. However, when they mixed fragments of their updates before sending them (FFL setting, right column of the figure), the server was only able to obtain noise instead of the original images.
\begin{figure}[htbp]
    \centering
    \begin{subfigure}{0.15\textwidth}
      \centering
      \includegraphics[width=\linewidth]{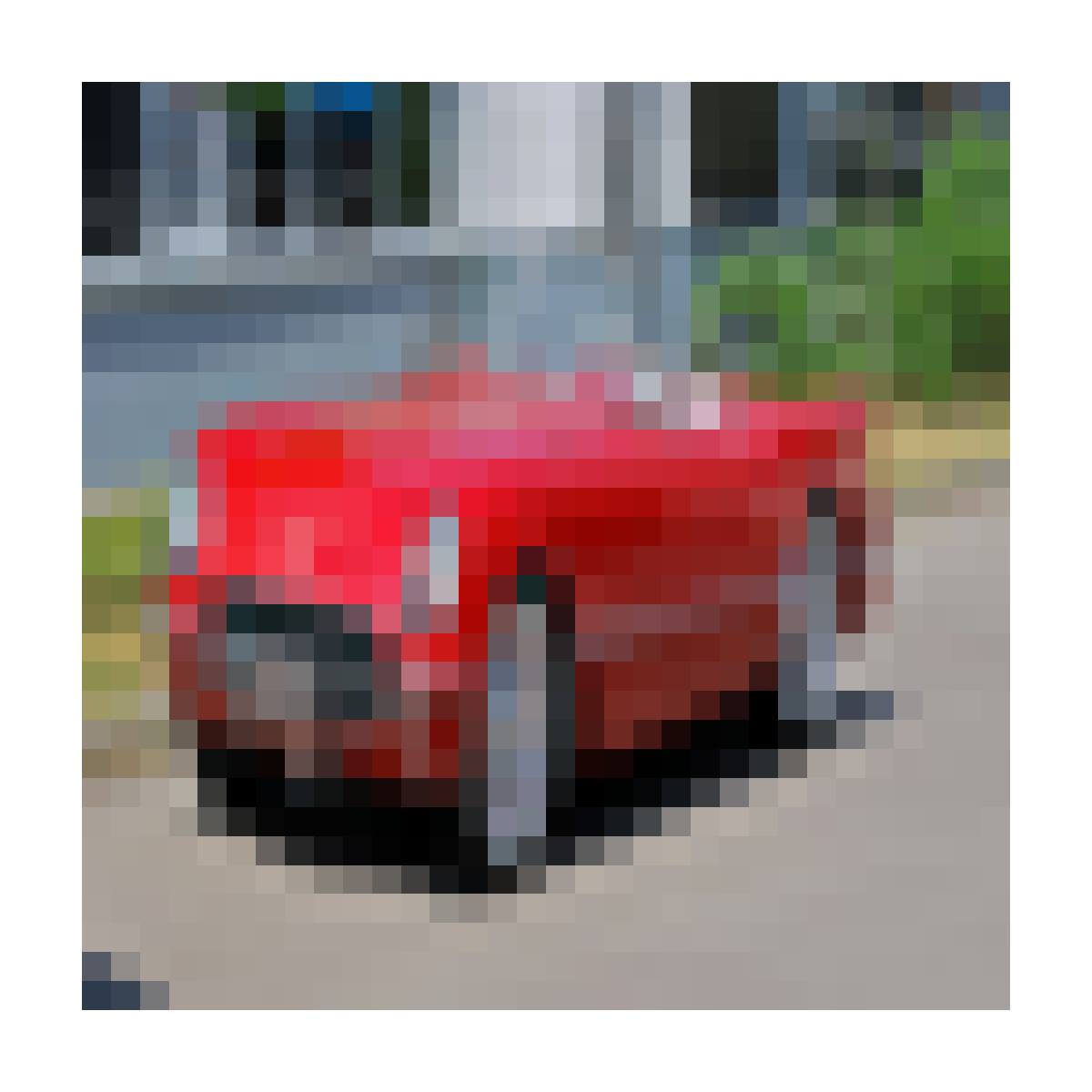}
    \end{subfigure}%
    \begin{subfigure}{0.15\textwidth}
      \centering
      \includegraphics[width=\linewidth]{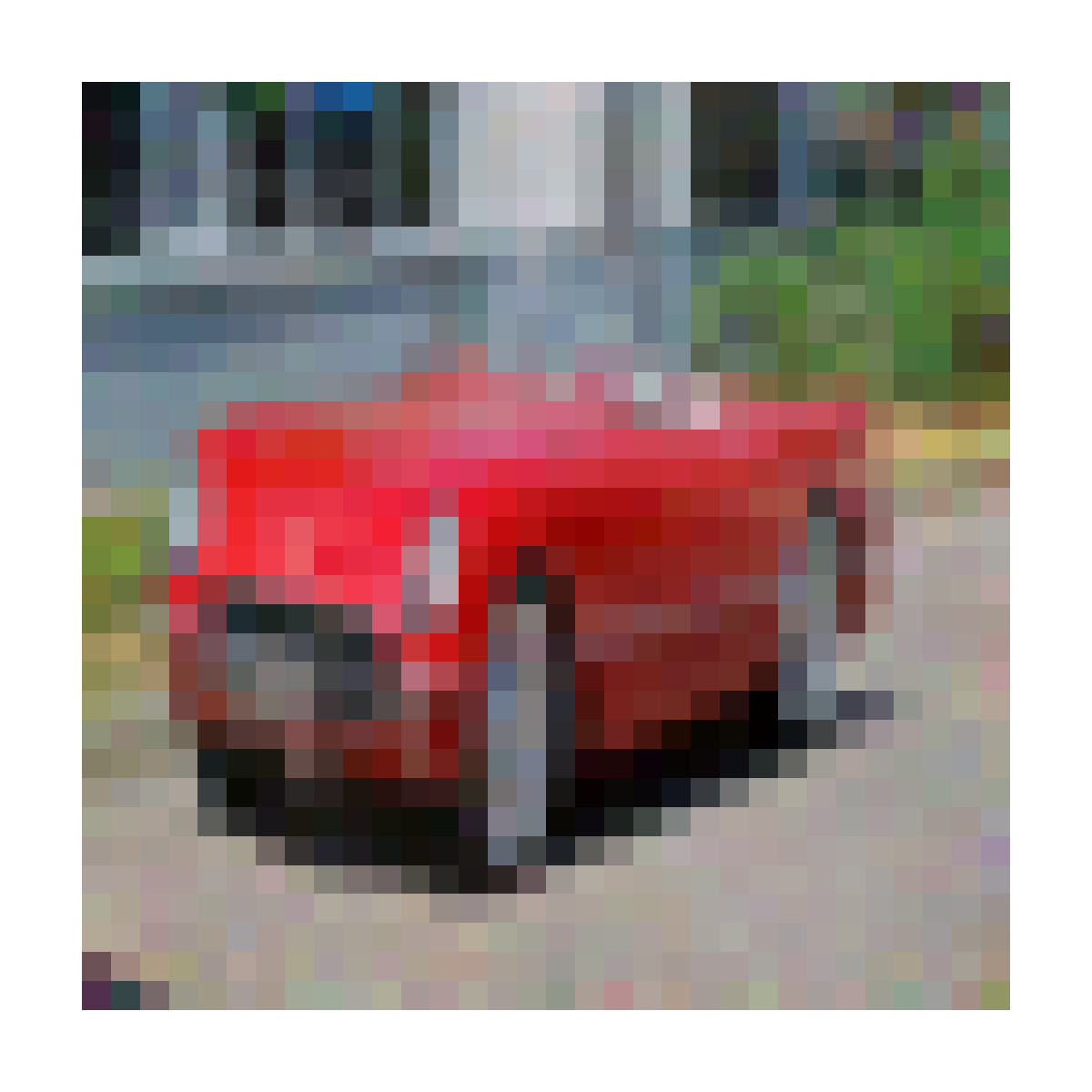}
    \end{subfigure}%
    \begin{subfigure}{0.15\textwidth}
      \centering
      \includegraphics[width=\linewidth]{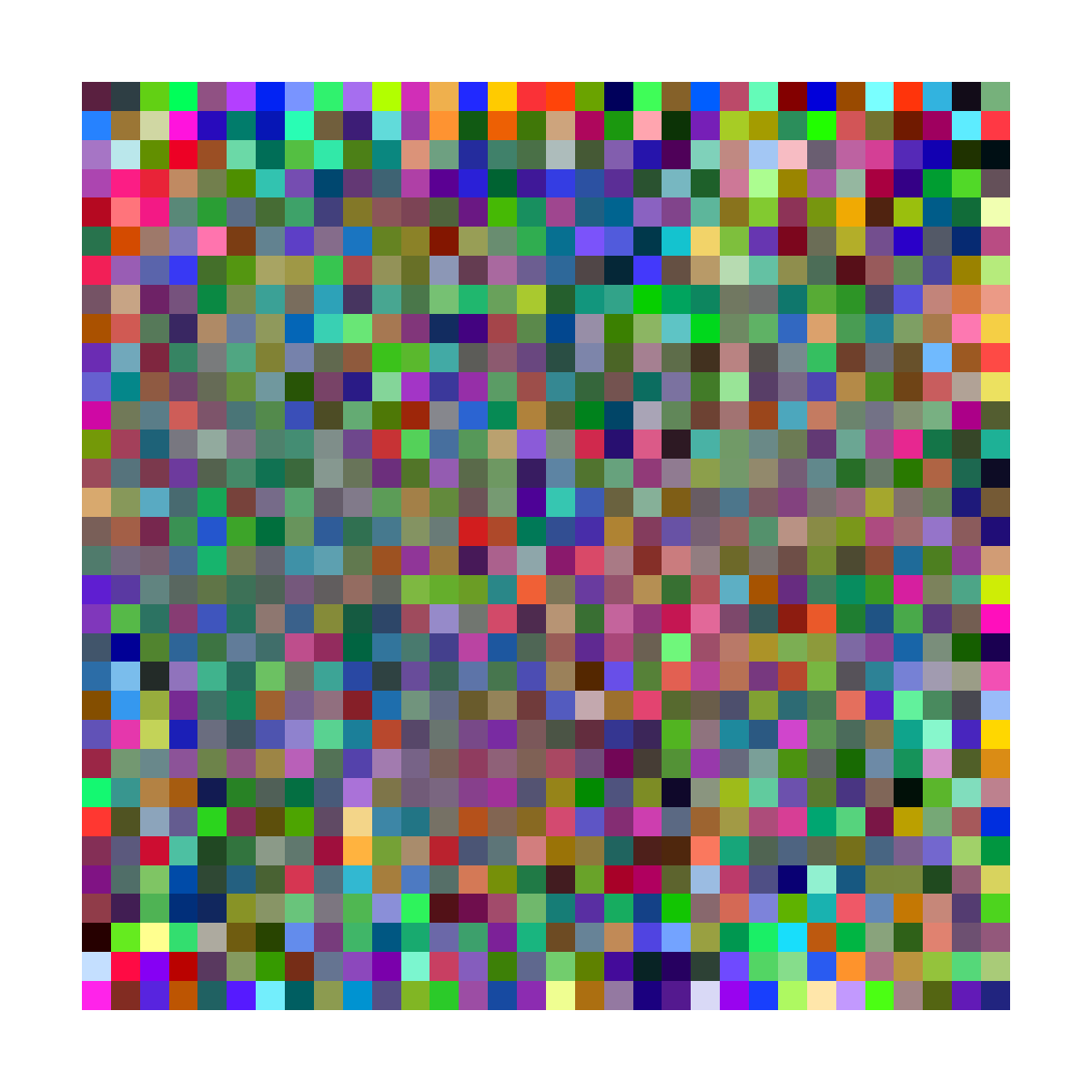}
    \end{subfigure}%
    \vspace{0.1\baselineskip}
    \begin{subfigure}{0.15\textwidth}
      \centering
      \includegraphics[width=\linewidth]{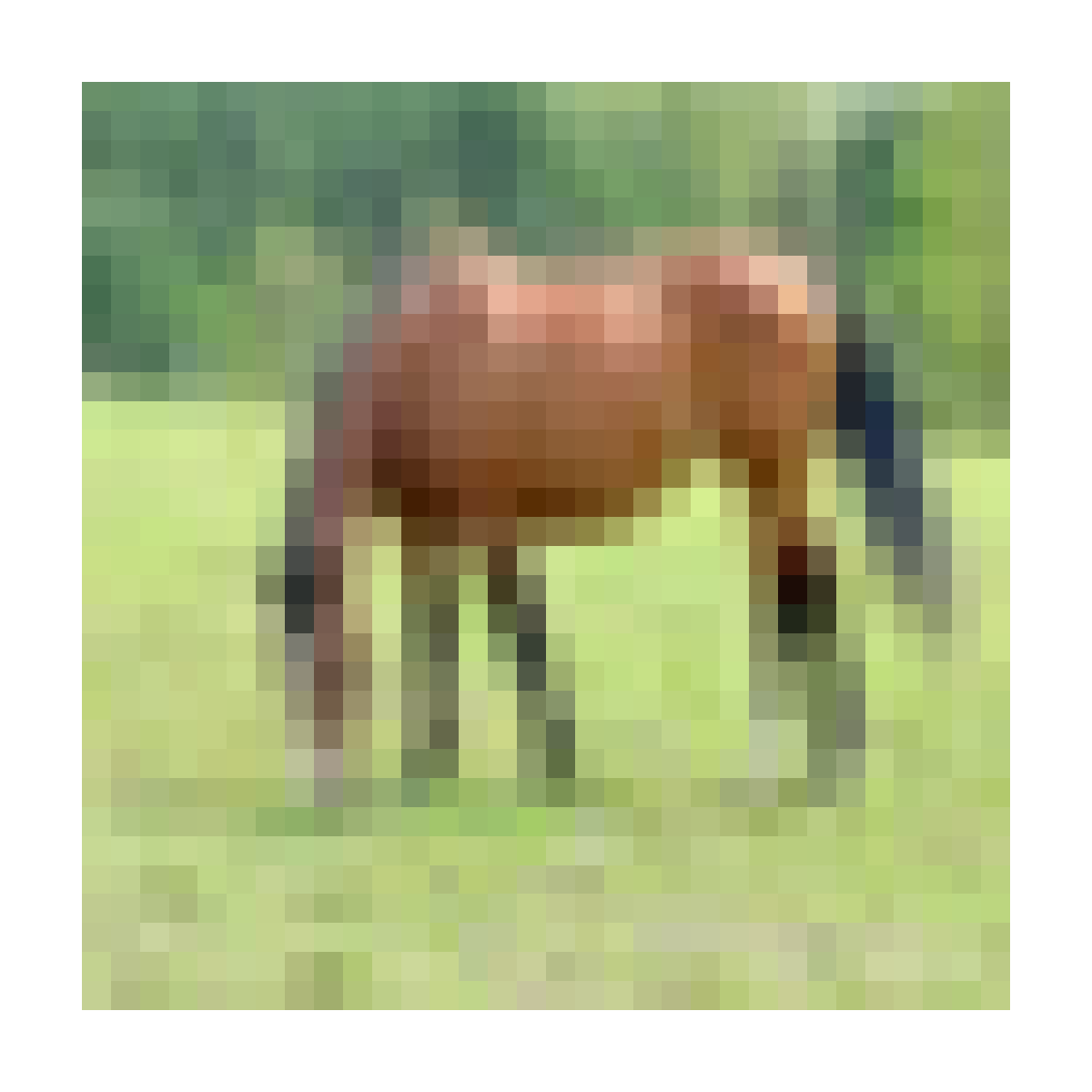}
    \end{subfigure}%
    \begin{subfigure}{0.15\textwidth}
      \centering
      \includegraphics[width=\linewidth]{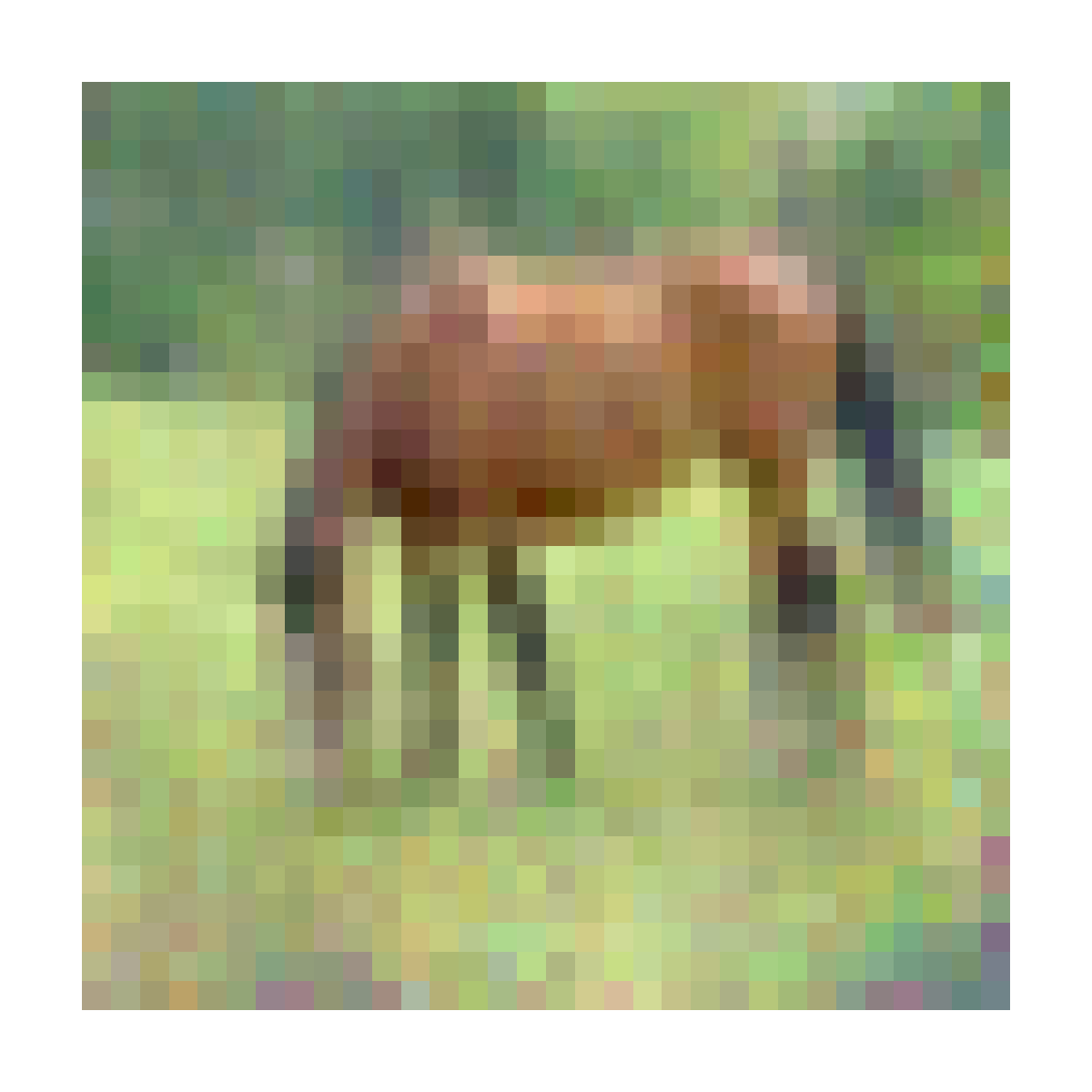}
    \end{subfigure}%
    \begin{subfigure}{0.15\textwidth}
      \centering
      \includegraphics[width=\linewidth]{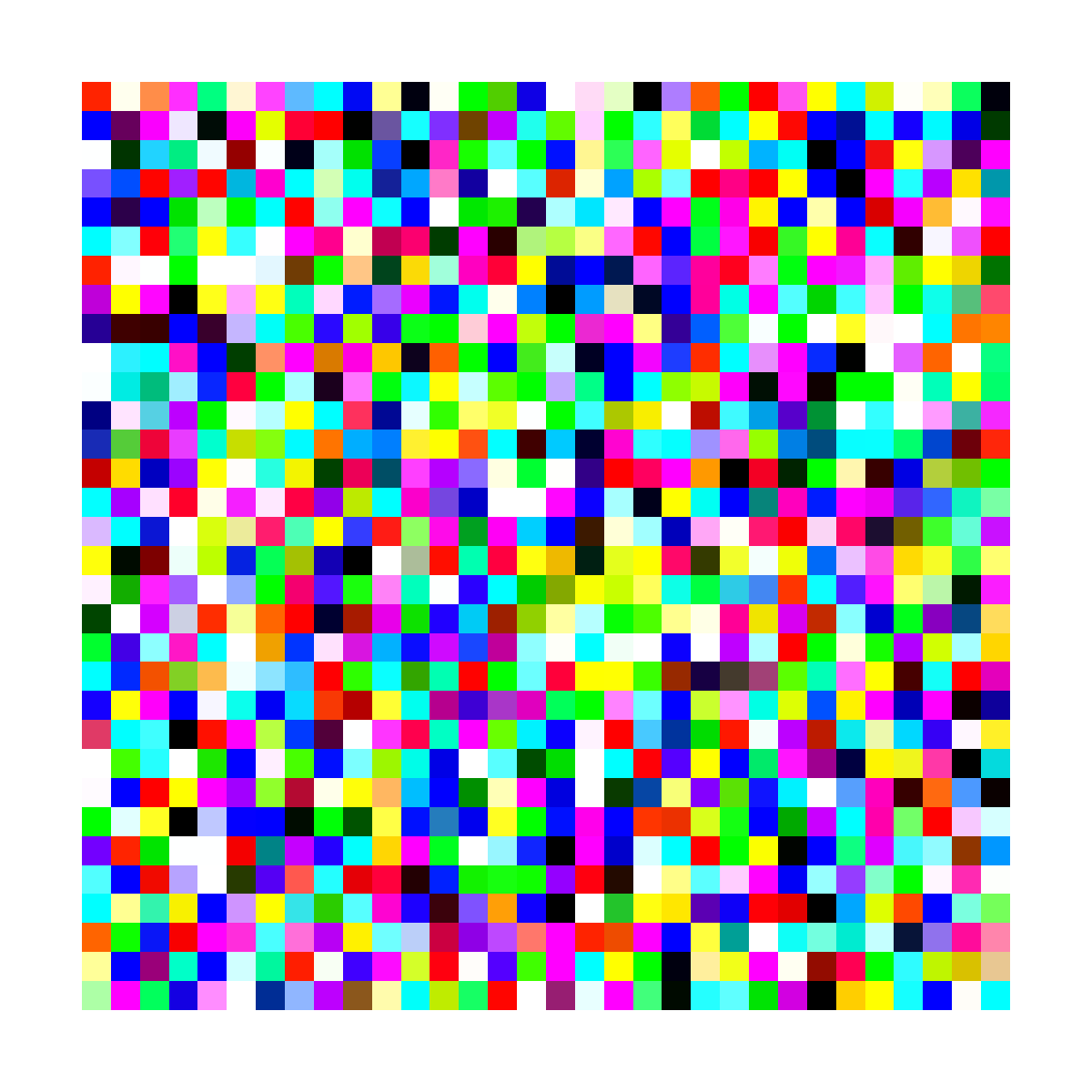}
    \end{subfigure}%
    \caption{Reconstruction of two input images from the gradients of two participants.  Left: Two input images. Middle: Reconstruction from original gradients (FL). Right: Reconstruction from mixed gradients (FFL).}
    \label{fig:one_image_construction}
\end{figure}

Similarly, Fig.~\ref{fig:batch_construction_peerkj} shows the results when two participants $k$ and $j$ sent their updates computed on a batch of $8$ images to the server. The figure exemplifies on two different batches of input images; the first example is given in the three upper rows and the second one in the three lower rows. In the FL setting, the server was able to recover a lot of information about the participants' training data. However, when they mixed their updates before sending them to the server (FFL setting), the latter just got a totally random batch of images.

\begin{figure}[htbp]
    \centering
    \begin{subfigure}{0.5\textwidth}
      \centering
      \includegraphics[width=\linewidth]{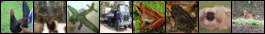}
    \end{subfigure}%
    \vspace{0.1\baselineskip}
    \begin{subfigure}{0.5\textwidth}
      \centering
      \includegraphics[width=\linewidth]{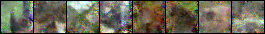}
    \end{subfigure}%
    \vspace{0.1\baselineskip}
    \begin{subfigure}{0.5\textwidth}
      \centering
      \includegraphics[width=\linewidth]{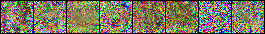}
    \end{subfigure}%
    \vspace{0.5\baselineskip}
    \begin{subfigure}{0.5\textwidth}
      \centering
      \includegraphics[width=\linewidth]{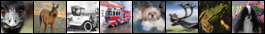}
    \end{subfigure}%
    \vspace{0.1\baselineskip}
    \begin{subfigure}{0.5\textwidth}
      \centering
      \includegraphics[width=\linewidth]{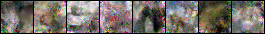}
    \end{subfigure}%
    \vspace{0.1\baselineskip}
    \begin{subfigure}{0.5\textwidth}
      \centering
      \includegraphics[width=\linewidth]{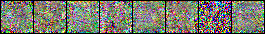}
    \end{subfigure}%
    \caption{The three upper rows show the reconstruction by the server of a batch of 8 input images based on the gradients of a participant. The first row shows the input private images, the second row shows the reconstructions from the FL original gradients of the participant, and the third row the reconstructions from the FFL mixed gradients of the participant. The three lower rows report another analogous example with a different batch of input images.}
    \label{fig:batch_construction_peerkj}
\end{figure}

\subsection{Runtime of FFL}
\label{sec:runtime}

\begin{table*}[t]
\scriptsize
\caption{CPU runtime in seconds of FFL in comparison with standard FL for one training round. Columns Dec. and Aggr. contain the decryption runtime of encrypted mixed updates, and the global reputation calculation and aggregation runtime on the server side (S). 
Column Mix. reports the overhead (extra runtime) on the participant side (P)
with respect to standard FL; this overhead results from local reputation calculation and the EXCHANGE\_FRAGMENTS protocol of FFL.}
\label{tab:runtime}
\resizebox{\textwidth}{!}{%
\begin{tabular}{|cl|ccc|ccc|ccc|ccc|}
\hline
\multicolumn{2}{|c|}{Benchmark}                                                                     & \multicolumn{3}{c|}{Adult-MLP}                                                            & \multicolumn{3}{c|}{MNIST-CNN}                                                            & \multicolumn{3}{c|}{CIFAR10-VGG16}                                                         & \multicolumn{3}{c|}{IMDB-BiLSTM}                                                           \\ \hline
\multicolumn{1}{|l|}{\multirow{2}{*}{\# of updates}} & \multicolumn{1}{c|}{\multirow{2}{*}{Method}} & \multicolumn{2}{c|}{S}                                               & P                  & \multicolumn{2}{c|}{S}                                               & P                  & \multicolumn{2}{c|}{S}                                                & P                  & \multicolumn{2}{c|}{S}                                                & P                  \\ \cline{3-14} 
\multicolumn{1}{|l|}{}                               & \multicolumn{1}{c|}{}                        & \multicolumn{1}{c|}{Dec.}               & \multicolumn{1}{c|}{Aggr.} & Mix.               & \multicolumn{1}{c|}{Dec.}               & \multicolumn{1}{c|}{Aggr.} & Mix.               & \multicolumn{1}{c|}{Dec.}               & \multicolumn{1}{c|}{Aggr.}  & Mix.               & \multicolumn{1}{c|}{Dec.}               & \multicolumn{1}{c|}{Aggr.}  & Mix.               \\ \hline
\multicolumn{1}{|c|}{\multirow{6}{*}{10}}            & FedAvg (FL)                                  & \multicolumn{1}{c|}{\multirow{5}{*}{0}} & \multicolumn{1}{c|}{0.001} & \multirow{5}{*}{0} & \multicolumn{1}{c|}{\multirow{5}{*}{0}} & \multicolumn{1}{c|}{0.001} & \multirow{5}{*}{0} & \multicolumn{1}{c|}{\multirow{5}{*}{0}} & \multicolumn{1}{c|}{0.064}  & \multirow{5}{*}{0} & \multicolumn{1}{c|}{\multirow{5}{*}{0}} & \multicolumn{1}{c|}{0.071}  & \multirow{5}{*}{0} \\ \cline{2-2} \cline{4-4} \cline{7-7} \cline{10-10} \cline{13-13}
\multicolumn{1}{|c|}{}                               & Median (FL)                                  & \multicolumn{1}{c|}{}                   & \multicolumn{1}{c|}{0.002} &                    & \multicolumn{1}{c|}{}                   & \multicolumn{1}{c|}{0.002} &                    & \multicolumn{1}{c|}{}                   & \multicolumn{1}{c|}{0.858}  &                    & \multicolumn{1}{c|}{}                   & \multicolumn{1}{c|}{0.816}  &                    \\ \cline{2-2} \cline{4-4} \cline{7-7} \cline{10-10} \cline{13-13}
\multicolumn{1}{|c|}{}                               & TMean (FL)                                   & \multicolumn{1}{c|}{}                   & \multicolumn{1}{c|}{0.003} &                    & \multicolumn{1}{c|}{}                   & \multicolumn{1}{c|}{0.007} &                    & \multicolumn{1}{c|}{}                   & \multicolumn{1}{c|}{1.201}  &                    & \multicolumn{1}{c|}{}                   & \multicolumn{1}{c|}{0.843}  &                    \\ \cline{2-2} \cline{4-4} \cline{7-7} \cline{10-10} \cline{13-13}
\multicolumn{1}{|c|}{}                               & MKrum (FL)                                   & \multicolumn{1}{c|}{}                   & \multicolumn{1}{c|}{0.004} &                    & \multicolumn{1}{c|}{}                   & \multicolumn{1}{c|}{0.013} &                    & \multicolumn{1}{c|}{}                   & \multicolumn{1}{c|}{0.822}  &                    & \multicolumn{1}{c|}{}                   & \multicolumn{1}{c|}{0.690}  &                    \\ \cline{2-2} \cline{4-4} \cline{7-7} \cline{10-10} \cline{13-13}
\multicolumn{1}{|c|}{}                               & FGold (FL)                                   & \multicolumn{1}{c|}{}                   & \multicolumn{1}{c|}{0.01}  &                    & \multicolumn{1}{c|}{}                   & \multicolumn{1}{c|}{0.017} &                    & \multicolumn{1}{c|}{}                   & \multicolumn{1}{c|}{0.497}  &                    & \multicolumn{1}{c|}{}                   & \multicolumn{1}{c|}{0.256}  &                    \\ \cline{2-14} 
\multicolumn{1}{|c|}{}                               & FFL                                          & \multicolumn{1}{c|}{0.064}              & \multicolumn{1}{c|}{0.021} & 0.231              & \multicolumn{1}{c|}{0.113}              & \multicolumn{1}{c|}{0.016} & 0.578              & \multicolumn{1}{c|}{0.372}              & \multicolumn{1}{c|}{2.308}  & 1.106              & \multicolumn{1}{c|}{0.321}              & \multicolumn{1}{c|}{1.773}  & 1.081              \\ \hline
\multicolumn{1}{|c|}{\multirow{6}{*}{50}}            & FedAvg (FL)                                  & \multicolumn{1}{c|}{\multirow{5}{*}{0}} & \multicolumn{1}{c|}{0.002} & \multirow{5}{*}{0} & \multicolumn{1}{c|}{\multirow{5}{*}{0}} & \multicolumn{1}{c|}{0.002} & \multirow{5}{*}{0} & \multicolumn{1}{c|}{\multirow{5}{*}{0}} & \multicolumn{1}{c|}{0.194}  & \multirow{5}{*}{0} & \multicolumn{1}{c|}{\multirow{5}{*}{0}} & \multicolumn{1}{c|}{0.275}  & \multirow{5}{*}{0} \\ \cline{2-2} \cline{4-4} \cline{7-7} \cline{10-10} \cline{13-13}
\multicolumn{1}{|c|}{}                               & Median (FL)                                  & \multicolumn{1}{c|}{}                   & \multicolumn{1}{c|}{0.007} &                    & \multicolumn{1}{c|}{}                   & \multicolumn{1}{c|}{0.008} &                    & \multicolumn{1}{c|}{}                   & \multicolumn{1}{c|}{6.843}  &                    & \multicolumn{1}{c|}{}                   & \multicolumn{1}{c|}{4.135}  &                    \\ \cline{2-2} \cline{4-4} \cline{7-7} \cline{10-10} \cline{13-13}
\multicolumn{1}{|c|}{}                               & TMean (FL)                                   & \multicolumn{1}{c|}{}                   & \multicolumn{1}{c|}{0.015} &                    & \multicolumn{1}{c|}{}                   & \multicolumn{1}{c|}{0.045} &                    & \multicolumn{1}{c|}{}                   & \multicolumn{1}{c|}{7.097}  &                    & \multicolumn{1}{c|}{}                   & \multicolumn{1}{c|}{4.267}  &                    \\ \cline{2-2} \cline{4-4} \cline{7-7} \cline{10-10} \cline{13-13}
\multicolumn{1}{|c|}{}                               & MKrum (FL)                                   & \multicolumn{1}{c|}{}                   & \multicolumn{1}{c|}{0.014} &                    & \multicolumn{1}{c|}{}                   & \multicolumn{1}{c|}{0.015} &                    & \multicolumn{1}{c|}{}                   & \multicolumn{1}{c|}{7.044}  &                    & \multicolumn{1}{c|}{}                   & \multicolumn{1}{c|}{5.657}  &                    \\ \cline{2-2} \cline{4-4} \cline{7-7} \cline{10-10} \cline{13-13}
\multicolumn{1}{|c|}{}                               & FGold (FL)                                   & \multicolumn{1}{c|}{}                   & \multicolumn{1}{c|}{0.042} &                    & \multicolumn{1}{c|}{}                   & \multicolumn{1}{c|}{0.035} &                    & \multicolumn{1}{c|}{}                   & \multicolumn{1}{c|}{1.609}  &                    & \multicolumn{1}{c|}{}                   & \multicolumn{1}{c|}{1.129}  &                    \\ \cline{2-14} 
\multicolumn{1}{|c|}{}                               & FFL                                          & \multicolumn{1}{c|}{0.263}              & \multicolumn{1}{c|}{0.04}  & 0.239              & \multicolumn{1}{c|}{0.509}              & \multicolumn{1}{c|}{0.033} & 0.583              & \multicolumn{1}{c|}{1.875}              & \multicolumn{1}{c|}{4.726}  & 1.188              & \multicolumn{1}{c|}{1.612}              & \multicolumn{1}{c|}{3.536}  & 1.174              \\ \hline
\multicolumn{1}{|c|}{\multirow{6}{*}{100}}           & FedAvg (FL)                                  & \multicolumn{1}{c|}{\multirow{5}{*}{0}} & \multicolumn{1}{c|}{0.004} & \multirow{5}{*}{0} & \multicolumn{1}{c|}{\multirow{5}{*}{0}} & \multicolumn{1}{c|}{0.004} & \multirow{5}{*}{0} & \multicolumn{1}{c|}{\multirow{5}{*}{0}} & \multicolumn{1}{c|}{0.341}  & \multirow{5}{*}{0} & \multicolumn{1}{c|}{\multirow{5}{*}{0}} & \multicolumn{1}{c|}{0.517}  & \multirow{5}{*}{0} \\ \cline{2-2} \cline{4-4} \cline{7-7} \cline{10-10} \cline{13-13}
\multicolumn{1}{|c|}{}                               & Median (FL)                                  & \multicolumn{1}{c|}{}                   & \multicolumn{1}{c|}{0.015} &                    & \multicolumn{1}{c|}{}                   & \multicolumn{1}{c|}{0.013} &                    & \multicolumn{1}{c|}{}                   & \multicolumn{1}{c|}{16.582} &                    & \multicolumn{1}{c|}{}                   & \multicolumn{1}{c|}{8.926}  &                    \\ \cline{2-2} \cline{4-4} \cline{7-7} \cline{10-10} \cline{13-13}
\multicolumn{1}{|c|}{}                               & TMean (FL)                                   & \multicolumn{1}{c|}{}                   & \multicolumn{1}{c|}{0.035} &                    & \multicolumn{1}{c|}{}                   & \multicolumn{1}{c|}{0.111} &                    & \multicolumn{1}{c|}{}                   & \multicolumn{1}{c|}{16.645} &                    & \multicolumn{1}{c|}{}                   & \multicolumn{1}{c|}{8.929}  &                    \\ \cline{2-2} \cline{4-4} \cline{7-7} \cline{10-10} \cline{13-13}
\multicolumn{1}{|c|}{}                               & MKrum (FL)                                   & \multicolumn{1}{c|}{}                   & \multicolumn{1}{c|}{0.028} &                    & \multicolumn{1}{c|}{}                   & \multicolumn{1}{c|}{0.031} &                    & \multicolumn{1}{c|}{}                   & \multicolumn{1}{c|}{27.875} &                    & \multicolumn{1}{c|}{}                   & \multicolumn{1}{c|}{16.893} &                    \\ \cline{2-2} \cline{4-4} \cline{7-7} \cline{10-10} \cline{13-13}
\multicolumn{1}{|c|}{}                               & FGold (FL)                                   & \multicolumn{1}{c|}{}                   & \multicolumn{1}{c|}{0.086} &                    & \multicolumn{1}{c|}{}                   & \multicolumn{1}{c|}{0.063} &                    & \multicolumn{1}{c|}{}                   & \multicolumn{1}{c|}{6.180}  &                    & \multicolumn{1}{c|}{}                   & \multicolumn{1}{c|}{2.273}  &                    \\ \cline{2-14} 
\multicolumn{1}{|c|}{}                               & FFL                                          & \multicolumn{1}{c|}{0.359}              & \multicolumn{1}{c|}{0.043} & 0.249              & \multicolumn{1}{c|}{1.01}               & \multicolumn{1}{c|}{0.035} & 0.594              & \multicolumn{1}{c|}{3.675}              & \multicolumn{1}{c|}{4.883}  & 1.319              & \multicolumn{1}{c|}{3.236}              & \multicolumn{1}{c|}{3.583}  & 1.282              \\ \hline
\end{tabular}%
}
\end{table*}

Table~\ref{tab:runtime} reports the CPU runtimes in seconds per training round for mixed update decryption, global reputation calculation and model aggregation on the server side, and local model computation, fragment exchanging and mixing on the participants' side. 
Note that the reported runtimes for exchanging and mixing fragments are the average for a single participant.
We computed the runtime for the three benchmarks with 10, 50 and 100 updates per training round to illustrate how the runtime scales with the number of updates.
On the server side,
we report the total runtime for all methods; in contrast,
on the participant side, we report 
the overhead (extra runtime)  with respect to standard FL. 

It can be seen that, as the number of updates and
the  model size increase, the FFL server's runtime 
grows less than for the other methods, which confirms 
our computation cost analysis
 of Appendix S.II-A in the supplementary materials. In particular, with CIFAR10-VGG16, 
FFL achieved the lowest runtime among all non-baseline methods. 
Furthermore, unlike FFL, the 
other methods are unable to thwart privacy attacks or provide adequate protection against security attacks. 

Regarding the runtime overhead incurred by each participant, the maximum runtime overhead resulting from exchanging and mixing fragments was 
about 1.319 seconds for the largest model we used, which was VGG16. Also, it is worth noting that the increase in the number of updates had little impact on the overhead of the participants because their computations 
essentially depend on the model dimensionality. 

Nevertheless, the FFL runtime is very small if we compare it with that of ~\cite{zhang2020batchcrypt}, which provides a level of privacy similar to ours but without being able to neutralize poisoned updates. Specifically, for an LSTM model containing only $4.02$ million parameters, ~\cite{zhang2020batchcrypt} report a runtime overhead for each participant of about $176$ seconds and a total runtime for the server of $174$ seconds to aggregate $50$ local updates. Note also that the hardware specifications employed by ~\cite{zhang2020batchcrypt} are superior to ours. 
Therefore, the participants in FFL can turn a blind eye to the computational overhead they incur in exchange for protecting their privacy, countering security attacks and learning a more accurate global model.

\section{Conclusions and Future Work}
\label{sec:conclusion}

In this paper, we have presented \emph{fragmented federated learning} (FFL), a novel approach based on cooperation between participants in FL systems to preserve their privacy without renouncing robust and accurate aggregation of their updates to the global model. FFL offers a practical solution where participants privately and efficiently exchange random fragments of their updates before sending them to the server. We have also proposed a novel reputation-based defense tailored for FFL that builds trust in the participants based on the quality of the mixed updates they send and the fragments they exchange.
We have demonstrated that the proposed framework can effectively counter privacy and security attacks.
All the above is achieved while obtaining a global model's accuracy similar to that of standard FL (when no attack is performed) and imposing affordable communication cost and computation overhead on the participating parties.  
The efficiency of FFL makes it applicable to large-scale FL systems.

As future work, we plan to evaluate the robustness of FFL against backdoor attacks. Also, we plan to test its performance when working with non-identical and independently distributed (non-iid) data.

\section*{Acknowledgments}

Thanks go to Ashneet Singh for help with some figures. This research was funded by the European Commission (projects H2020-871042 ``SoBigData++'' and H2020-101006879 ``MobiDataLab''), the Government of Catalonia (ICREA Acad\`emia Prizes to J.
Domingo-Ferrer, and D. S\'anchez, FI grant to N. Jebreel), and MCIN/AEI/ 10.13039/501100011033 and ``ERDF A way of making Europe'' under grant PID2021-123637NB-I00 ``CURLING''. The authors are with the UNESCO Chair in Data Privacy, but the views in this paper are their own and are not necessarily shared by UNESCO.

\ifCLASSOPTIONcaptionsoff
  \newpage
\fi

\bibliography{my_bib.bib}
\bibliographystyle{IEEEtran}

\begin{IEEEbiography}[{\includegraphics[width=1in,height=1.25in,clip,keepaspectratio]{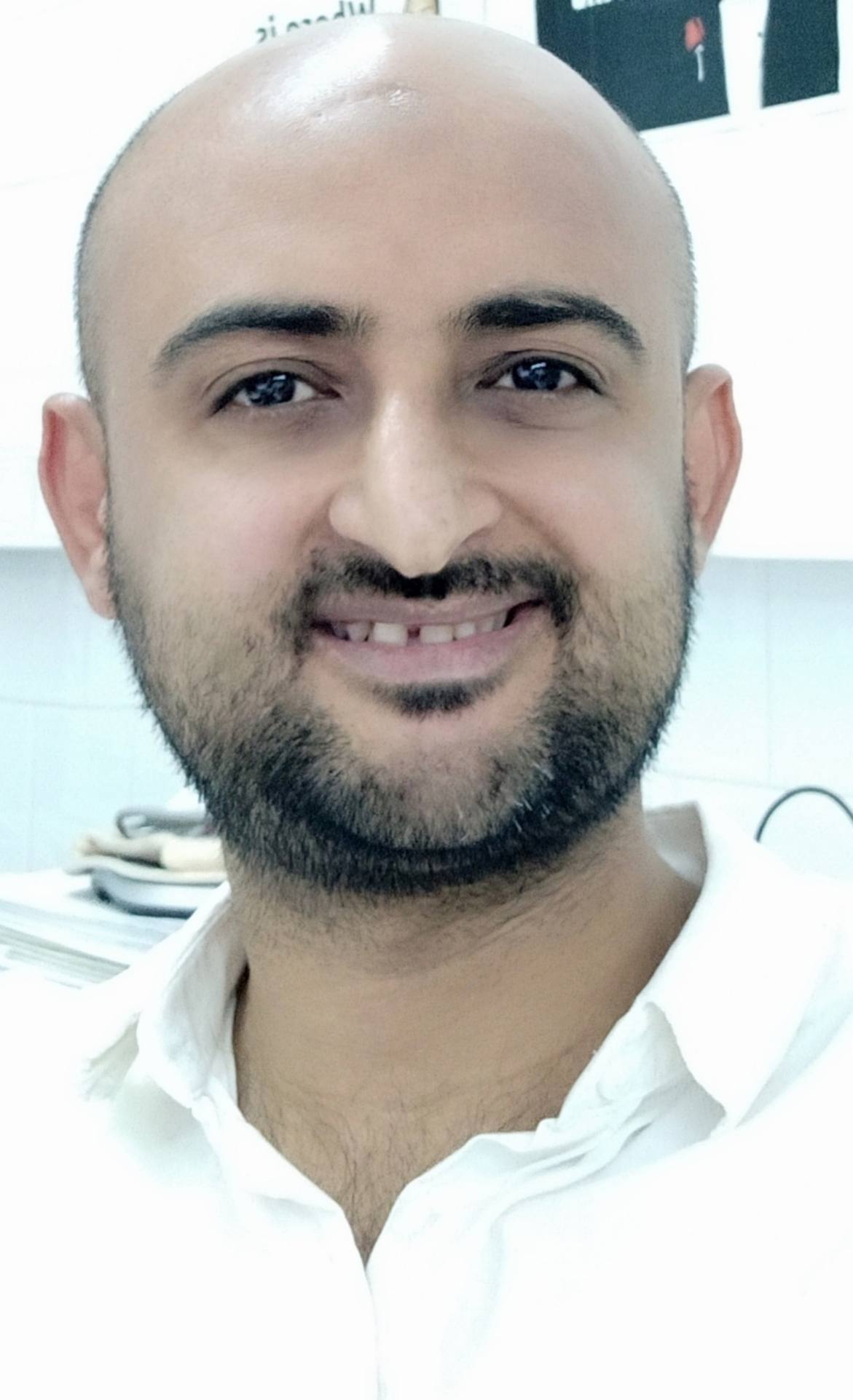}}]{Najeeb Moharram Jebreel} is a Ph.D. student in Computer Science at Universitat Rovira i Virgili, Tarragona, Catalonia. He obtained a B.Sc. in Computer Science from Hodeidah University and an M.Sc. in AI and Computer Security from Universitat Rovira i Virgili. His interests are in reconciling accuracy, security and privacy in distributed machine learning.
\end{IEEEbiography}

\begin{IEEEbiography}[{\includegraphics[width=1in,height=1.25in,clip,keepaspectratio]{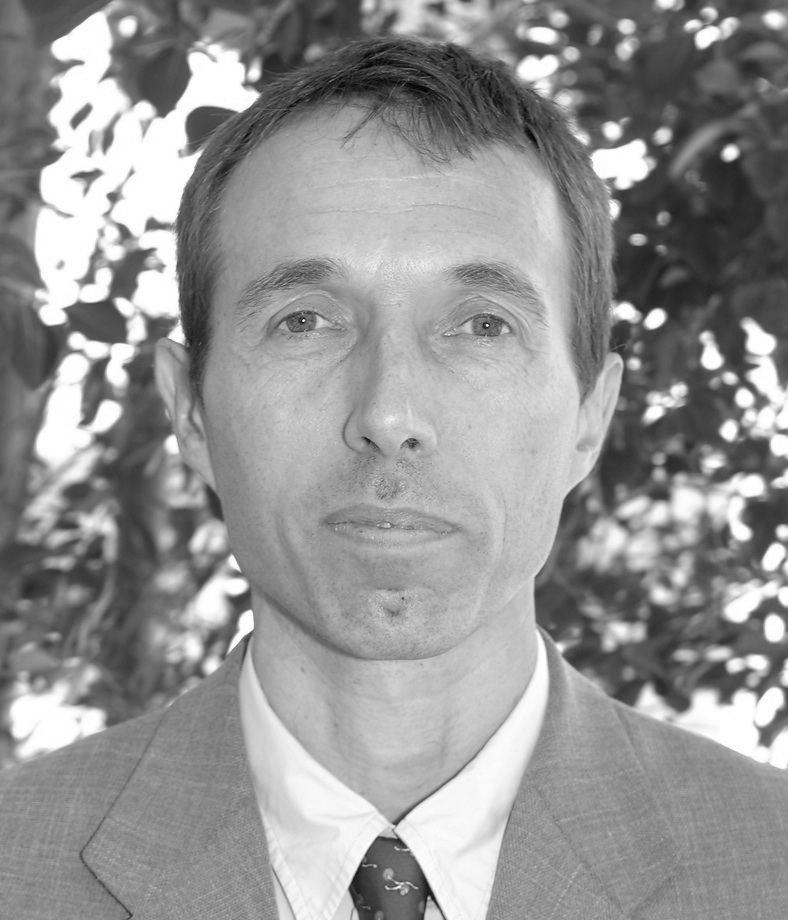}}]{Josep Domingo-Ferrer}
(Fellow, IEEE)
is a Distinguished Professor of Computer Science and
an ICREA-Acad\`emia Researcher at Universitat Rovira i Virgili,
Tarragona, Catalonia, where he holds the UNESCO Chair in Data Privacy
and leads CYBERCAT. He received B.Sc.-M.Sc. and Ph.D. degrees in Computer Science from the Autonomous University of Barcelona, and an B.Sc.-M.Sc. in Mathematics from U.N.E.D. His research interests are in security, privacy and ethics-by-design computing.
\end{IEEEbiography}

\begin{IEEEbiography}[{\includegraphics[width=1in,height=1.25in,clip,keepaspectratio]{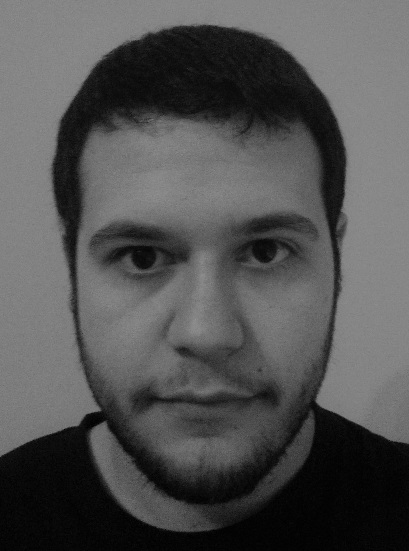}}]{Alberto Blanco-Justicia}
is a senior postdoc at Universitat Rovira i Virgili, Tarragona,
Catalonia. He obtained a M.Sc. in Computer Security and a Ph.D. in Computer Engineering and Mathematics, both from Universitat Rovira i Virgili. His interests are in data privacy, security, ethically-aligned design and machine learning explainability. He has been involved in several European R+D projects and industrial contracts.
\end{IEEEbiography}

\begin{IEEEbiography}[{\includegraphics[width=1in,height=1.25in,clip,keepaspectratio]{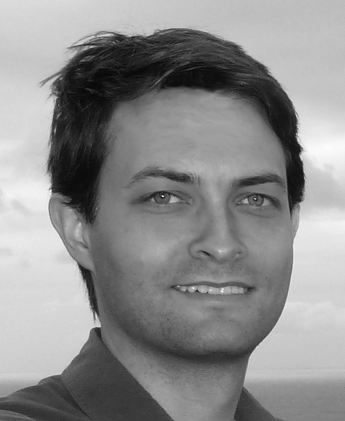}}]{David S{\'a}nchez}
is a Serra H\'unter Full Professor of Computer Science and an ICREA-Acad\`emia Researcher at Universitat Rovira i Virgili, Tarragona,
Catalonia. He received his Ph.D. in Computer Science from the Technical University of Catalonia. His research interests include data semantics, ontologies, machine learning and data privacy.
\end{IEEEbiography}

\appendices
\section{Convergence Analysis}
\label{sec:convergence}
To prove the convergence of FFL, first we need to prove that applying the adaptive aggregation in Protocol 1 on the mixed updates gives the same result as applying FedAvg on the original updates when all participants are honest.
\begin{proposition}
\label{props4}
Given that parameters move among participants but remain in their original coordinates when mixing updates, applying the adaptive aggregation in Protocol 1 on the mixed updates when all participants are honest produces the same result as applying FedAvg on the original updates.
\end{proposition}
\begin{proof}
\label{proof4}
 Let $g_c:\mathbb{R}^n \mapsto \mathbb R$ be the coordinate-wise average operator which aggregates each coordinate independently from the others, 
 and $(W_{1, i}, \ldots, W_{n,i})$ be the $i$-th coordinate parameters for $n$ local updates. Let $\pi:\mathbb{R}^n \mapsto \mathbb{R}^n$ be a random  permutation function that is applied to the $i$-th coordinate parameters
 to yield one of $n!$ possible permutations.
 Given the permutation-insensitivity property of $g_c$, it
 holds that 
 \begin{equation}
 \label{eq_avg_insensitivity}
     g_c(W_{1, i}, \ldots, W_{n,i}) =  g_c(\pi(W_{1, i}, \ldots, W_{n,i}))
 \end{equation}
 for any permutation function $\pi$. Hence, applying $g_c$ on mixed updates will produce the same result as applying it on the original updates as long as the exchanged parameters remain in their original coordinates.
 
 FedAvg is a weighted average operator that aggregates the parameters coordinate-wise proportionally to the number of data points of their senders. 
 Let us assume the server gives a trust score of 1 to every sender as they are honest. 
 Given that each participant $k$ weights her original local update by her number of data points $d_k$, the adaptive aggregation process in Protocol 1 will produce the same global model as applying FedAvg on the original local updates.
\end{proof}

\begin{proposition}
\label{props5}
Given the process in Protocol 1, the convergence of FFL is guaranteed as long as FedAvg converges.
\end{proposition}

\begin{proof}
\label{proof5}
Given Proposition~S.\ref{props4}, we just need to prove that FFL fulfills two conditions: 1) it removes poisoned mixed updates from model aggregation and 2) it does not modify the parameters of good mixed updates in the adaptive aggregation phase.

\textbf{Condition 1:} As the training evolves, the attackers who exchange poisoned fragments or send poisoned updates are expected to obtain low global and local reputations and thus not be selected by the server or honest participants for future training or fragments exchange. Therefore, poisoned mixed updates are expected to be removed from the global model aggregation.

\textbf{Condition 2:} As the training evolves, participants with global reputations greater than $Q1_{\gamma}$ are expected to be honest and thus be 
repeatedly selected by the server for future training rounds.
This will make their global reputations greater and greater, and thus their trust values will converge to 1 due to the use of the $\tanh$ functions. 
That is,
\begin{equation}
\label{lim_honest}
   \Lim{t \to \infty} \tanh{((\gamma^t_k|\gamma^t_k > Q1_{\gamma^t}) - Q1_{\gamma^t})} = 1.
\end{equation}
As a result, the parameters of the mixed updates 
of those participants will stay unaltered.
Since honest participants above $Q1_\gamma$ are about $75\%$ of 
the $K$ participants in the system, the average of their mixed updates 
is expected to be an unbiased estimator of the average of the $K$ participants' original updates when all participants are honest. 
Let $W^\infty_*$ be the global model obtained 
when all participants are honest, 
and let $(W^\infty_k)_{mix}$ and $W^\infty_k$ be, respectively,
participant $k$'s mixed and original updates, after $\infty$
iterations. 
Let $H$ be the set of expected honest participants with global reputations greater than $Q1_{\gamma^t}$.
Then we have
\begin{equation}
\begin{aligned}
\label{lim_avg}
\Lim{t \to \infty} \frac{1}{\sum_{k \in H} d_k} \sum_{k \in H}  \nu_k^t (W_{k}^{t})_{mix} 
\simeq \frac{1}{\sum_{k \in H} d_k} \sum_{k \in H}  (W_{k}^\infty)_{mix} \\
\simeq \frac{1}{\sum_{k \in H} d_k} \sum_{k \in H}  d_k W_{k}^\infty 
\simeq  W^\infty_*.
\end{aligned}
\end{equation}
This shows that Condition 2 is satisfied.
\end{proof}

\section{Complexity analysis}
\label{sec:complexity}
In this appendix, we analyze the computation and communication cost overheads of FFL. 
\subsection{Computation overhead of FFL}
\label{sec:computation}
We examine the computation overhead added by FFL to standard FL for the server and participants, and we compare it with related works in the state of 
the art.
We also consider the last-layer dimensionality $D_L$ when looking 
at FFL and FoolsGold~\cite{fung2020limitations}. 
Note that the last layer usually contains a very small number of parameters compared to the whole model parameters and thus has a much lower dimensionality. 
For example, in the VGG16 model we use in our experiments, the model contains about $15$ million parameters, whereas its last layer contains about $250K$ parameters.

\textbf{Server side.} The computation overhead on the server side 
can be broken down into the following parts:
(1) selecting participants for training based on their global reputations, which costs $\mathcal{O}(n)$;
(2) decrypting $n$ encrypted seeds and their corresponding $n$ OTP-encrypted mixed updates, which costs $\mathcal{O}(nc + nD)$, with $c$ being a constant; 
(3) extracting the gradients of the mixed updates, computing their magnitudes, computing the median of the computed magnitudes and the distances between the computed magnitudes and their median, which costs $\mathcal{O}(nD + nD + n\log{n} + n)$,
(4) computing the median of the last-layer gradients of $n$ mixed updates 
and the cosine similarity between the last-layer gradients to their median, which costs $\mathcal{O}(n\log{n}D_L + nD_L)$; and
(5) computing the similarity vector $sim^t$, updating the global reputation vector and finally computing the trust vector $\nu^t$, which 
costs $\mathcal{O}(n)$.
Therefore, the overall computation overhead of the server is $\mathcal{O}(n(D+\log{n} +D_L\log{n} + D_L + 1))$.

\textbf{Participant side.} The participant's computation overhead can be broken down into the following parts:
(1) computing the first quartile of her local reputation vector using the Quickselect algorithm, which costs $\mathcal{O}(n)$;
(2) computing the $2048$-bit Diffie-Hellman power, which costs $\mathcal{O}(1)$;
(3) fragmenting the local update, encrypting two updates using the two generated OTPs, and subtracting one OTP-encrypted update from the other, which costs $\mathcal{O}(D + 2D + D)$; and
(4) encrypting a $3072$-bit seed and finally updating the local reputation of a single participant in a single index, which costs $\mathcal{O}(1 + 1)$.
Therefore, the overall computation overhead of a participant is $\mathcal{O}(n + D)$.

Now, let us compare the computation overhead of FFL with that of 
the following methods: BREA~\cite{so2020byzantine}, the median~\cite{yin2018byzantine}, the trimmed mean~\cite{yin2018byzantine}, multi-Krum~\cite{blanchard2017machine}, and FoolsGold~\cite{fung2020limitations}.
Note that, while BREA~\cite{so2020byzantine} tries to counter both privacy and security attacks, the other methods are only meant to counter security attacks.

\begin{table*}[!ht]
\centering
\caption{Comparison of computation overhead}
\label{tab:computation_overhead}
\begin{tabular}{|l|c|c|}
\hline
Framework & Server(s)                     & Participant    \\ \hline
FFL        & $\mathcal{O}(n(D+\log{n}+D_L\log{n}+D_L+1))$          & $\mathcal{O}(n + D)$          \\ \hline
BREA       & $\mathcal{O}((n^2 + nD)\log^2{n\log{\log{n}}})$             & $\mathcal{O}(nD + n^2)$ \\ \hline
Median      & $\mathcal{O}(n\log{n}D)$ &  0 \\ \hline
Trimmed mean     & $\mathcal{O}(n\log{n}D)$                  &        0  \\ \hline
Multi-Krum    &    $\mathcal{O}(n^2D)$    &   0 \\ \hline
FoolsGold    &    $\mathcal{O}(n^2 D_L)$    &  0  \\ \hline
\end{tabular}
\end{table*}

Table~S.\ref{tab:computation_overhead} shows that FFL imposes less computation overhead on the server than the other methods.
If we compare the computation overhead of FFL with that of BREA on both the server and the participant, FFL is much more efficient and thus more suitable for large-scale FL applications.

\subsection{Communication cost of FFL}
\label{sec:communication}

Here we turn to the total communication cost of FFL with respect to the model size $D$ and the number of participants $n$. These two parameters are 
relevant because $D$ has a high impact on the size 
of the transmitted messages whereas $n$ has a high impact 
on their number. 
We also use $c$ to denote a constant communication cost
for any sent or received parameter with a constant size,
such as the encrypted seeds.

\textbf{Server side.} The communication cost at the server side 
can be broken down into the following parts:
(1) sending one global model to each participant and receiving one OTP-encrypted mixed update from each participant, which costs $\mathcal{O}(2nD)$; 
and (2) sending a $32$-bit global reputation $\gamma_k$ to each participant and receiving a $3072$-bit encrypted seed with each OTP-encrypted mixed update, which costs $\mathcal{O}(\sim 3n Kbit) = \mathcal{O}(nc)$.
Therefore, the overall communication cost of the server is $\mathcal{O}(2nD + nc)$.

\textbf{Participant side.} The participant's communication cost can be broken 
down into the following parts:
(1) receiving one global model from the server and sending one OTP-encrypted mixed update to the server, which costs $\mathcal{O}(2D)$; 
(2) sending two OTP-encrypted updates and receiving two OTP-encrypted updates, which costs $\mathcal{O}(4D)$; and
(3) exchanging two $2048$-bit Diffie-Hellman powers (see~\cite{rfc5114} for justification of this 2048 length) and two $3072$-bit encrypted seeds (see~\cite{barker2006recommendation} for justification of this $3072$ length) with another participant, plus sending $3072$-bit encrypted seed to the server and receiving a $32$-bit global reputation $\gamma_k$ from it, which costs $\mathcal{O}(6c)$.
Therefore, the overall communication cost for each participant is $\mathcal{O}(6D + 6c)$.

Now, let us compare the communication cost of FFL with that of each of the standard FL, BREA~\cite{so2020byzantine}, PEFL~\cite{privacy_enhanced9524709} and ShieldFL~\cite{ma2022shieldfl}.
Table~S.\ref{tab:communication_overhead} shows the communication complexity of each framework on the server(s) and the participant.

\begin{table}[!ht]
\centering
\caption{Comparison of communication cost, where $\delta$ is the message expansion factor resulting from HE, and $c$ is a constant communication cost
for any sent or received parameter with a constant size, such as the encrypted seeds.}
\label{tab:communication_overhead}
\begin{tabular}{|l|c|c|}
\hline
Framework & Server(s)                     & Participant    \\ \hline
FL        & $\mathcal{O}(2nD)$                        & $\mathcal{O}(2D)$          \\ \hline
FFL       & $\mathcal{O}(2nD + nc)$             & $\mathcal{O}(6D + 6c)$ \\ \hline
BREA      & $\mathcal{O}(2nD + n^2c)$ & $\mathcal{O}(2nD + D + n)$ \\ \hline
PEFL     & $\mathcal{O}(4 \delta nD + \delta D)$                  & $\mathcal{O}(\delta D + D)$          \\ \hline
ShieldFL    &    $\mathcal{O}(13 \delta nD)$    & $\mathcal{O}(2 \delta D)$    \\ \hline
\end{tabular}
\end{table}

In standard FL, the server sends $n$ copies of the global model to $n$ participants and receives $n$ local updates back, which costs $\mathcal{O}(2nD)$.
FFL imposes the same cost on the server plus $\mathcal{O}(nc)$, which corresponds to the received $n$ encrypted seeds, and the sent $n$ global reputations.
Since the server receives the encrypted seeds along with her encrypted mixed updates in the same messages and sends the $32$-bit global reputations along with the updated global models in the same messages, FFL sends and receives the same number of messages as in standard FL.
Moreover, the server only incurs 3072 bits $\approx 0.0004$ MB communication overhead for every received mixed update and $32$ bits for every sent updated global model, which is negligible compared to the sizes of the current state-of-art deep learning models. 
For example, the sizes of the three deep learning models we use in our experimental analysis below are $0.0874$ MB (for a CNN model), $58.9131$ MB (for a VGG model), and $50.6453$ MB (for a BiLSTM model).
Therefore, FFL imposes almost the same communication cost on the server as standard FL.

On the participant's side, in FFL each participant sends and receives 6 updates in total and also incurs $6c$ additional cost, as note above.
Therefore, the number of messages a participant needs to send or receive is 6 in total, which is three times as many as in standard FL.
Nevertheless, the communication cost each participant incurs is reasonable in exchange for the benefit she receives in protecting her privacy and learning a more accurate global model, and also in comparison with the other works.
 
To put these values in context, let us compare them with the communication cost of BREA~\cite{so2020byzantine}, PEFL~\cite{privacy_enhanced9524709} and ShieldFL~\cite{ma2022shieldfl}, three state-of-the-art frameworks that provide a similar level of privacy and accuracy as FFL. 

In BREA~\cite{so2020byzantine}, the server sends and receives the same number of updates. 
However, it incurs an additional quadratic communication cost from 
receiving the locally computed Euclidean distances from the $n$ participants.
This quadratic term can make the server a bottleneck of communication as the number of participants in the system increases. 
FFL, however, imposes a much lower communication cost on the server, as we have shown.
On the participant side, the communication cost of BREA depends on both the number of participants $n$ and the model dimensionality $D$. 
In real-world FL scenarios, we expect both $n$ and $D$ to be large, and hence a significant communication cost will be imposed on the participant.
In contrast, the communication cost on the participant side in FFL is independent of the number of participants in the system.

Both PEFL~\cite{privacy_enhanced9524709} and ShieldFL~\cite{ma2022shieldfl} are HE-based frameworks that expand updates by some factor $\delta$ (which is usually large) and thus add a necessarily significant communication and computation overhead on the server and participants.
PEFL~\cite{privacy_enhanced9524709} tries to reduce the size of an encrypted update by using a packing technique, which results in roughly a threefold 
message expansion, as the authors claim.
This makes the server's communication cost about $\mathcal{O}(12nD + 3D)$ and the participant's cost about $\mathcal{O}(4D)$. Note that the packing and HE cause a high computation overhead on the participants' devices.
Moreover, the two non-colluding servers in these two frameworks exchange encrypted updates and some other parameters between them several times so that they can filter the poisoned updates while protecting privacy, which involves exchanging several large size messages.

To summarize, FFL imposes almost the same communication cost on the server as standard FL (without privacy preservation), and only a marginally higher 
communication cost on participants. This makes it 
more suitable for real-world privacy-preserving FL than
the state-of-the-art methods.

To confirm the theoretical communication analysis above, we next provide some actual communication times that may illustrate the feasibility of our approach's (and its practical benefits w.r.t. related works).
To this end, we have considered the average speeds of mobile internet communication among countries
worldwide~\footnote{\url{https://www.speedtest.net/global-index}. Checked on Aug. 1, 2022.}.
Specifically, the current average download speed, upload speed and latency are, respectively, $31.01$ megabits per second ($Mbps$), $8.66$ $Mbps$ and $29$ milliseconds ($ms$).

According to these figures, we next report the communication and latency times of FFL for the MNIST-CNN and the CIFAR1010-VGG16 benchmarks, and we compare them with those of standard FL and BREA~\cite{so2020byzantine}.
Comparing with BREA makes sense because it is, like our framework, a single-server framework and it also involves exchanging messages between participants.
We were not able to compare with PEFL~\cite{privacy_enhanced9524709} or ShieldFL~\cite{ma2022shieldfl} because they do not give the exact value of the used expansion factor $\delta$ and also because they are two-server frameworks.
In this setting, the sizes of the CNN and VGG16 models are $0.0874MB$ ($0.6992Mbit$) and $58.9131MB$ ($471.3048Mbit$), which represent small and large DL models.
Also, we consider $100K$ participants in the system, which would correspond to a large-scale FL system.

Communication times are based on the following constraints:
\begin{itemize}
    \item When a participant receives (downloads) a message from the server, we consider her/his download speed.
    \item When a participant sends a message to the server, we consider her/his upload speed.
    \item When two participants exchange messages, we consider their upload speed. 
\end{itemize}

We next report communication times for each framework's server-participant and participant-participant interaction for one training round.

\textbf{FL communication time}. In standard FL, each participant downloads a copy of the global model from the server and uploads her/his local update to the server, so there is no participant-participant communication.
Thus, the server-participant communication time for the MNIST-CNN is $0.6992Mbit/31.01Mbps + 0.6992Mbit/8.66Mbps  \approx 0.1033$ seconds. 
By adding one latency for downloading and one for uploading, 
we obtain that the total
communication time is
$0.1033 + 2 \times 29/1000 \approx 0.1613$ seconds.
The communication time of CIFAR10-VGG16 can be computed following the same steps, which result in a total communication time of about $69.6797$ seconds.

\textbf{FFL communication time}. 
In FFL, there are server-participant and participant-participant communications.
In the server-participant communication, each participant downloads from the server a copy of the global model along with a $32$-bit global reputation value in the same message.
Also, each participant uploads one OTP-encrypted mixed update to the server along with a $3,072$ bit encrypted seed in the same message.
Thus, the server-participant communication time for MNIST-CNN is 
$(0.6992Mbit + 0.000032Mbit)/31.01Mbps + (0.6992Mbit + 0.003072Mbit)/8.66Mbps \approx 0.1036$ seconds.
In the participant-participant communication, each participant sends one $2,048$-bit Diffie-Hellman power and receives one $2,048$-bit Diffie-Hellman power.
Also, s/he sends two OTP-encrypted updates plus a $3,072$ bit encrypted seed and receives the same.
Thus, the participant-participant communication time is 
$(4 \times 0.6992Mbit + 2 \times 0.002048Mbit + 2\times 0.003072Mbit)/8.66Mbps \approx  0.3241$ seconds.
By adding the latencies, we get a 
total communication time for MNIST-CNN of about 
$0.1036 + 2 \times 29/1000 + 0.3241 + 4 \times 29/1000 = 0.6018$ seconds.
The communication time of CIFAR10-VGG16 can be computed following the same steps, which result in a total 
communication time of about $287.4900$ seconds.
Therefore, the FFL communication time overheads over the standard FL for MNIST-CNN and CIFAR10-VGG16 
are $273.09\%$ and $312.59\%$, respectively.

\textbf{BREA communication time}. In BREA~\cite{so2020byzantine}, there are server-participant and participant-participant communications.
In the server-participant communication, each participant downloads a copy of the global model from the server.
Also, each participant uploads to the server a $100K$ floating-point vector (the locally computed Euclidean distances) and then uploads the aggregated shares. We consider the vector to be a $32$ bit floating-point.
Thus, the server-participant communication time for MNIST-CNN is $(0.6992Mbit)/31.01Mbps + (0.6992Mbit + 0.000032Mbit \times 100,000)/8.66Mbps \approx 0.4728$ seconds.
In the participant-participant communication, each participant sends $100K$ shares and receives $100K$ shares.
Thus, the participant-participant communication time is $(2 \times 0.6992Mbit \times 100,00)/8.66Mbps \approx 16,147.8060$ seconds.
By adding latencies, we get a total communication
for MNIST-CNN of about $0.4728 + 3 \times 29/1000 + 16,147.8060 + 200,000 \times 29/1000 = 21,948.3658$ seconds.
The communication time for CIFAR10-VGG16 can be computed following the same steps, which result in a total communication time of about $10,890,507.4916$ seconds.
Therefore, the BREA communication time overheads over standard FL for MNIST-CNN and CIFAR10-VGG16 are $13,607,070.37\%$ and $15,629,283.44\%$, respectively. 

These figures, summarized in Table~S.\ref{tab:communication_time}, show that the communication time imposed by FFL is significantly lower than that of the BREA, which provides similar security and privacy to FFL. 

Note that we do not include the time required to establish 
communications, which is expected to be small compared to the message-dependent times we consider.

\begin{table*}[t]
\centering
\caption{Comparison of communication time in seconds for one training round for the MNIST-CNN and CIFAR10-VGG16 benchmarks. 
S-P indicates the server-participant communication time. P-P indicates the participant-participant communication time.}
\label{tab:communication_time}
\resizebox{\textwidth}{!}{%
\begin{tabular}{|l|cccc|cccc|}
\hline
\multirow{2}{*}{\begin{tabular}[c]{@{}l@{}}Framework/\\ Benchmark\end{tabular}} & \multicolumn{4}{c|}{MNIST-CNN}                                                                                 & \multicolumn{4}{c|}{CIFAR10-VGG16}                                                                                      \\ \cline{2-9} 
                                                                               & \multicolumn{1}{c|}{S-P}    & \multicolumn{1}{c|}{P-P}         & \multicolumn{1}{c|}{Latency}    & Total time  & \multicolumn{1}{c|}{S-P}     & \multicolumn{1}{c|}{P-P}             & \multicolumn{1}{c|}{Latency}    & Total time      \\ \hline
FL                                                                              & \multicolumn{1}{c|}{0.1033} & \multicolumn{1}{c|}{0}           & \multicolumn{1}{c|}{0.0580}     & 0.1613      & \multicolumn{1}{c|}{69.6217} & \multicolumn{1}{c|}{0}               & \multicolumn{1}{c|}{0.0580}     & 69.6797         \\ \hline
FFL                                                                             & \multicolumn{1}{c|}{0.1036} & \multicolumn{1}{c|}{0.3241}      & \multicolumn{1}{c|}{0.1740}     & 0.6018      & \multicolumn{1}{c|}{69.6220} & \multicolumn{1}{c|}{217.6939}        & \multicolumn{1}{c|}{0.1740}     & 287.4900        \\ \hline
BREA                                                                            & \multicolumn{1}{c|}{0.4728} & \multicolumn{1}{c|}{16,147.8060} & \multicolumn{1}{c|}{5,800.0870} & 21,948.3658 & \multicolumn{1}{c|}{69.9912} & \multicolumn{1}{c|}{10,884,637.4134} & \multicolumn{1}{c|}{5,800.0870} & 10,890,507.4916 \\ \hline
\end{tabular}%
}
\end{table*}

\end{document}